\documentclass[letterpaper]{article}
\usepackage[preprint]{aaai2027}
\usepackage[hyphens]{url}
\usepackage{graphicx}
\usepackage{natbib}
\usepackage{caption}
\usepackage{amsmath,amssymb,amsthm}
\usepackage{algorithm}
\usepackage{algpseudocode}
\usepackage{booktabs}
\usepackage{array}
\usepackage{afterpage}

\newtheorem{proposition}{Proposition}
\newtheorem{theorem}{Theorem}
\newtheorem{corollary}{Corollary}
\newtheorem{lemma}{Lemma}

\title{MATERO-RCA: Mode-Aware Trajectory-Level Energy-Based Root-Set Optimization for Industrial Root Cause Analysis}
\author{
    Chengyu Tao\textsuperscript{\rm 1}\corresponding,
    Chunxi Huang\textsuperscript{\rm 2},
    Runquan Xiao\textsuperscript{\rm 3}
}
\affiliations{
    \textsuperscript{\rm 1}College of Mechanical and Vehicle Engineering, Hunan University, China\\
    \textsuperscript{\rm 2}School of Design, Hunan University, China\\
    \textsuperscript{\rm 3}Xiamen Xiangyu Group Corporation, China
}

\begin{document}

\maketitle

\begin{abstract}

Root cause analysis (RCA) for contextual anomalies in industrial time series is
challenging because responses depend jointly on control commands, operating
states, and coupled physical variables. A response can appear
marginally normal yet violate its operating context. Events may involve
multiple roots and alarms, with each root assigned an observation-only
effect confined to its recorded trajectory or a physical-propagation effect on
descendants. We propose Mode-Aware Trajectory-Level Energy-Based Root-Set
Optimization for Root Cause Analysis (MATERO-RCA), which jointly optimizes a
root set, root-effect modes, and auxiliary counterfactual trajectories.
Its graph-wide objective combines alarm resolution with temporal compatibility
across local causal relations. A Temporal Compatibility Network
    (CompatNet) maps parent-conditioned trajectory likelihoods to calibrated
    compatibility energies. A Counterfactual Repair Network (RepairNet) initializes
    mode-aware counterfactual trajectories for objective-directed gradient
refinement. An exact mixed-integer linear program minimizes a residual-cover
lower bound, enabling certified best-bound search over the finite admissible
root--mode space under the fixed inner solver. Experiments on simulated and real industrial datasets demonstrate superior RCA performance over representative baselines.
\end{abstract}

\section{Introduction}

Industrial root cause analysis (RCA) seeks the variables that initiate an
alarmed temporal event. Industrial systems couple event-driven control with physical
processes. Commands, operating states, and measured responses are temporally
coupled. After a command change, a delayed response may appear normal in
isolation yet violate the expected command--response relation and trigger an
alarm. Such a mismatch indicates a contextual anomaly
\citep{chandola2009anomaly,hussain2023discovering,mehling2026enabling}.
Time-series anomaly detection (AD) models such context through recent history
and cross-variable dependencies
\citep{carmona2022neural,deng2021graph,li2021multivariate}, but identifies
anomalous windows or relations rather than the root causes of alarms.

RCA in industrial systems remains challenging for three reasons.
(i)~\emph{Unknown root-effect modes}: a root may have an \emph{observation-only}
effect confined to its recorded trajectory or 
a \emph{physical-propagation} effect in which a physical change at the root affects
descendants through their normal causal mechanisms.
(ii)~\emph{Multi-root, multi-alarm events}: an individual alarm may be jointly
caused by multiple roots, while a single event may contain multiple concurrent
alarms.
(iii)~\emph{Implicit temporal dynamics}: a variable-level causal graph identifies
causal relations but not how their trajectories evolve. Industrial command--state
responses can be state-dependent, stochastic, and multimodal, making an explicit structural causal model (SCM) difficult to specify or learn. Figure~\ref{fig:motivation}
illustrates these challenges in one toy industrial event.

\begin{figure}[t]
    \centering
    \includegraphics[width=\columnwidth]{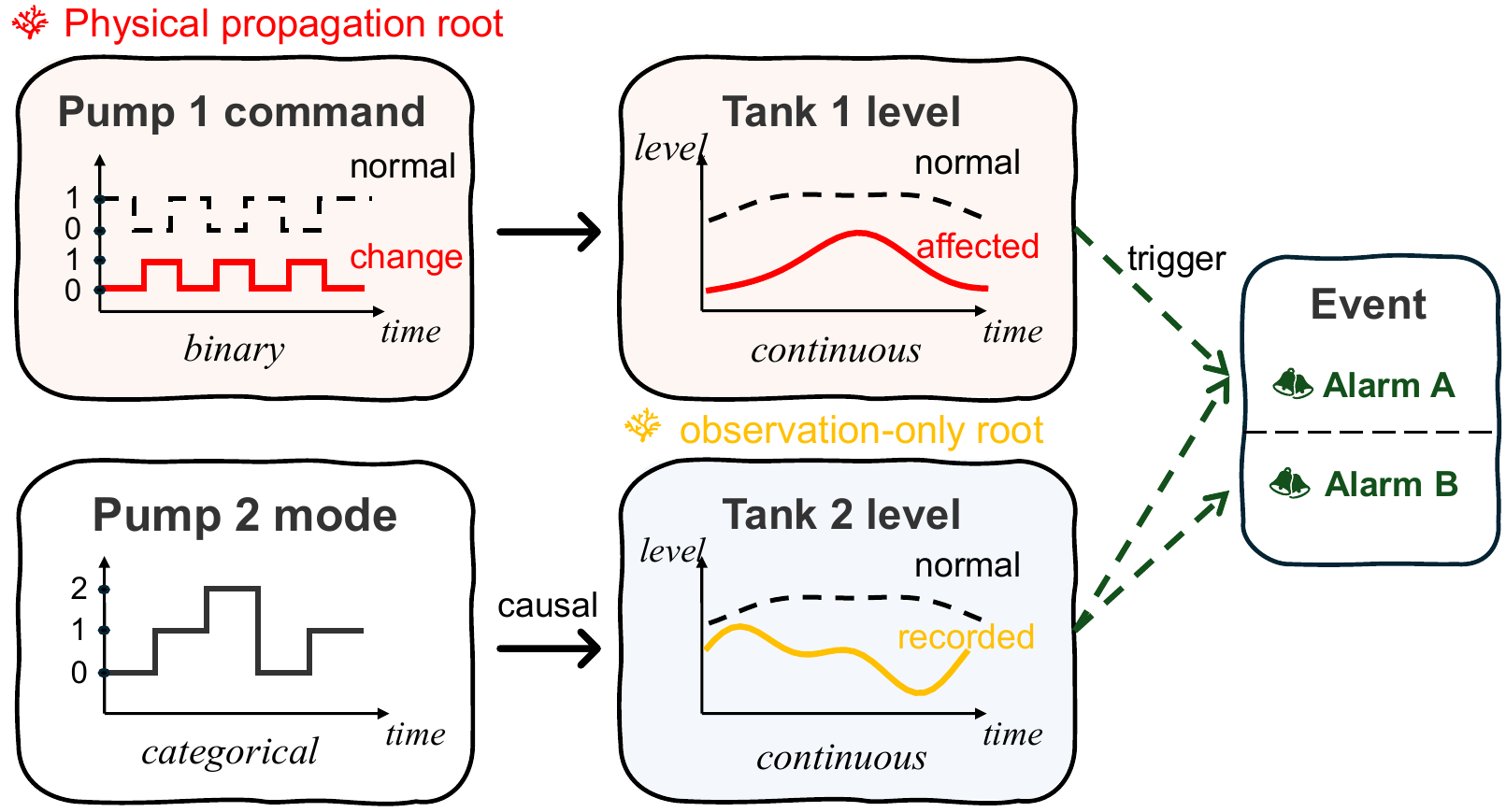}
    \caption{Illustrative industrial event in which roots with physical-propagation
    and observation-only effects jointly trigger one of multiple alarms.}
    \label{fig:motivation}
\end{figure}

Temporal RCA remains limited. EasyRCA~\citep{assaad2023root} and
T-RCA~\citep{zan2024fly} construct root sets from preidentified anomalous
subgraphs using graph structure and anomaly onset times, with EasyRCA
additionally testing direct-effect shifts. This dependence on detected
anomalies can exclude marginally plausible contextual roots before RCA. 
AERCA~\citep{han2025root} instead ranks inferred exogenous deviations.
These methods assume normal downstream propagation and do not model
observation-only root effects.
Dynamic counterfactual RCA~\citep{weilbach2024counterfactual} and
NetCause~\citep{chraim2026netcause} rank candidates by their forward-simulated
effects on subsequent system behavior rather than directly optimizing
counterfactual trajectories to resolve the observed alarm. Among general RCA
methods, CALI~\citep{suhr2026root} considers related root-effect distinctions
only for i.i.d.\ outliers rather than temporal events.

To address these limitations, we propose mode-aware trajectory-level energy-based root-set
optimization for root cause analysis (MATERO-RCA). MATERO-RCA jointly selects a
root set and its root-effect modes, using auxiliary counterfactual trajectories
to resolve active alarms and restore temporal compatibility across local causal
relations. Unlike rollout-based methods
\citep{weilbach2024counterfactual,chraim2026netcause}, these trajectories are
optimized directly under the unified objective.
Our contributions are summarized as follows:
\begin{itemize}
    \item We propose MATERO-RCA for alarmed temporal events in industrial systems,
    jointly optimizing an unknown-cardinality root set, root-effect modes, and
    auxiliary counterfactual trajectories.

    \item We develop an alarm-directed, graph-factored objective with a root-set cardinality penalty. A Temporal Compatibility Network (CompatNet) supplies calibrated
    relation-wise and alarm-context energies by learning conditional or joint likelihoods over mixed industrial signals.
    Low energies indicate alarm resolution and temporal compatibility. A Counterfactual Repair Network (RepairNet) initializes
    mode-aware counterfactual trajectories for objective-directed gradient refinement under the
    same objective.    
    
    \item We derive a sound residual-cover lower bound and an exact
    mixed-integer linear program (MILP) that minimizes it during best-bound search.
    The resulting solver prioritizes competitive root sets and certifies
    fixed-oracle outer optimality and $\epsilon$-optimal alternatives within the
    finite admissible root--mode space.

\end{itemize}

\section{Related Work}

\begingroup
\setlength{\parskip}{0pt}
\noindent\textbf{Temporal Anomaly Detection and RCA.}
Time-series AD models contextual normality from temporal and cross-variable
structure. NCAD~\citep{carmona2022neural} contrasts suspect and context windows.
GDN~\citep{deng2021graph} learns cross-sensor dependency graphs, while
InterFusion~\citep{li2021multivariate} embeds inter-metric and temporal
dependence. They detect anomalous windows or variables rather than roots.
Temporal RCA remains sparse. EasyRCA~\citep{assaad2023root} and
T-RCA~\citep{zan2024fly} construct root sets from detected anomalous subgraphs
using graph structure and onset order. EasyRCA further compares direct effects
across regimes, whereas T-RCA traverses a threshold-based summary causal graph.
Such anomaly-first filtering can exclude marginally plausible contextual roots
before RCA. AERCA~\citep{han2025root} learns lagged Granger dependencies and normal
exogenous distributions, then scores inferred exogenous deviations. Dynamic
counterfactual RCA~\citep{weilbach2024counterfactual} replaces candidate
node--time mechanisms with normal counterparts in an abducted dynamic SCM and
ranks candidates by the induced return to normality. NetCause~\citep{chraim2026netcause}
removes candidate anomalous transitions from a graph-temporal world model and
ranks their downstream effects. These methods use local temporal evidence or
candidate-wise simulations rather than jointly optimizing a trajectory-level
root-set explanation.

\par\noindent\textbf{General Causal RCA.}
Beyond temporal approaches, general RCA localizes roots using anomaly scores,
mechanism changes, or interventions. BARO~\citep{pham2024baro} ranks variables
by Bayesian change-point scores, whereas SmoothTraversal~\citep{orchard2025root}
traverses marginal anomaly-score shifts derived from conditional models under
polytree and single-root assumptions. CIRCA~\citep{li2022causal} uses
parent-conditioned residuals, while Causal outlier
attribution~\citep{budhathoki2022causal} and RootCLAM~\citep{han2023on} score
exogenous-noise contributions. Moreover, distinguishing between normal and abnormal local mechanisms,
RCD~\citep{ikram2022root} and RCG~\citep{ikram2025root} test dependence on a
failure indicator, while StableRCA~\citep{lin2026stablerca} detects
Markov-boundary shifts. These methods assume non-root symptoms preserve normal
mechanisms, excluding observation-only root effects. Set-aware methods accommodate multiple roots but
return per-root rankings or scores for queried sets
\citep{nagalapatti2025robust,lohse2026prim,lu2026probability}. Finally,
CALI~\citep{suhr2026root} uses latent-variable mixture inference to distinguish
root-effect modes, with guarantees restricted to sparse independent hard interventions and
polynomial Gaussian SCMs that do not cover temporally coupled
industrial trajectories.
\par
\endgroup

\section{Methodology}

\begin{figure*}[t]
    \centering
\includegraphics[width=0.85\textwidth]{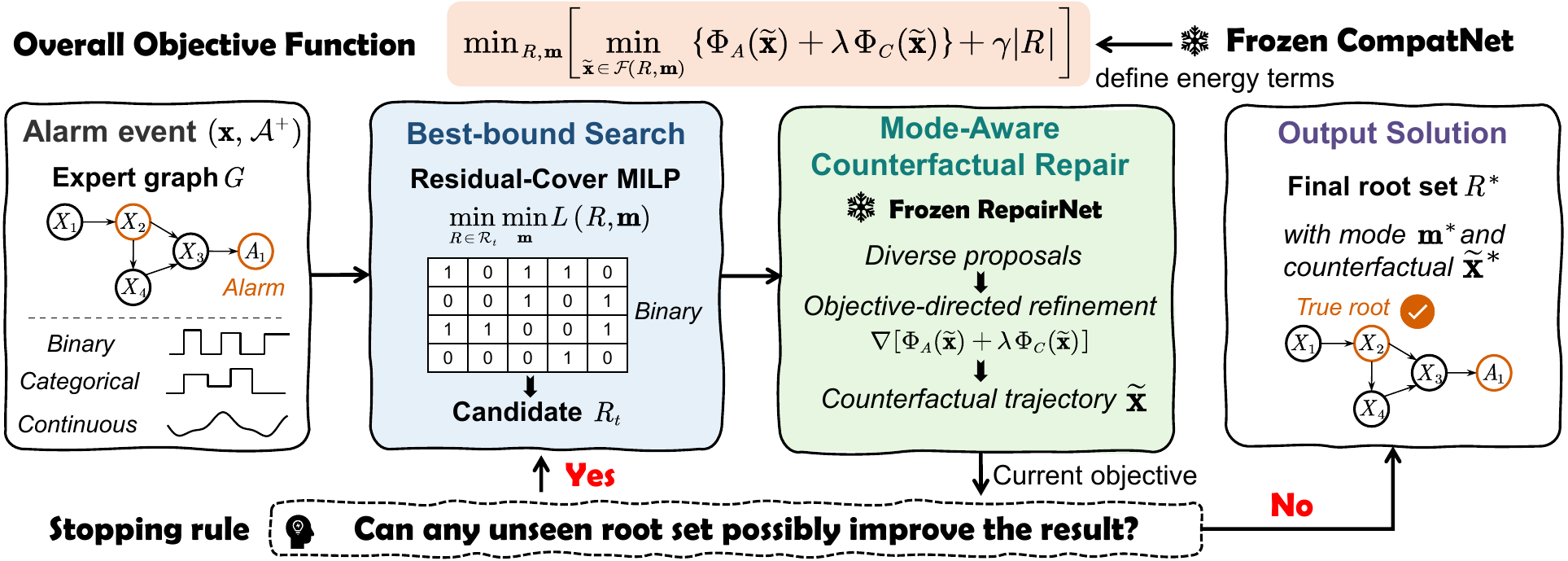}
\caption{MATERO-RCA integrates graph-factored temporal compatibility and
mode-aware counterfactual trajectory optimization into certified
residual-cover best-bound search for root-set optimization.}
    \label{fig:framework}
\end{figure*}

\subsection{Problem Setting and Notation}

Let $G$ be a provided causal graph over $d$ observed temporal variables (signals)
$\mathcal X=\{X_i\}_{i=1}^d$ and $q$ alarm indicators
$\mathcal A=\{A_j\}_{j=1}^q$. $G$ encodes variable-level causal relations
without lag-specific edges or time-unrolled nodes. All variable trajectories and
alarm-parent contexts use the same event window $t=0,\ldots,T-1$. Each $X_i$ is
represented by the trajectory
$\mathbf x_i=(x_{i,0},\ldots,x_{i,T-1})^\top\in\Omega_i^T$, where $\Omega_i$
is its binary, categorical, or continuous domain. Normal
training windows are $\mathcal D_N=\{\mathbf x^{(n)}\}_{n=1}^{N}$. At test time,
an alarmed temporal event is represented by $(\mathbf x,\mathcal A^+)$, where
$\mathbf x=(\mathbf x_1,\ldots,\mathbf x_d)$ contains its variable trajectories
and $\mathcal A^+\subseteq\mathcal A$ its active alarms. Temporal relations are
learned from $\mathcal D_N$ without known structural equations or exogenous-noise
distributions.

For $v\in\mathcal X\cup\mathcal A$, $\operatorname{pa}(v)\subseteq\mathcal X$
denotes its causal parents. For $A_j\in\mathcal A^+$, let
$\operatorname{an}(A_j)\subseteq\mathcal X$ denote its variable ancestors, and
let $\mathcal A^\top\subseteq\mathcal A^+$ contain the active alarms with no
active alarm descendants. Since non-ancestors cannot causally reach a top-level
active alarm, candidates are restricted to
$\mathcal C=\bigcup_{A_j\in\mathcal A^\top}\operatorname{an}(A_j)$.
Every top-level active alarm must be covered, giving
$\mathcal R_0=\{R\subseteq\mathcal C:1\leq|R|\leq K,\,
R\cap\operatorname{an}(A_j)\neq\varnothing,\, \forall A_j\in\mathcal A^\top\}$.
For $R\in\mathcal R_0$, let
$\mathbf m=(m_r)_{r\in R}\in\{\mathrm{o},\mathrm{p}\}^{|R|}$ collect its
auxiliary root-effect modes. Under $\mathrm{o}$, a root's effect is confined to
its recorded trajectory. Under $\mathrm{p}$, the root effect must propagate
through normal-compatible local causal relations toward the alarm-parent
context of at least one top-level active alarm. These modes are auxiliary and
need not be uniquely identifiable.
\emph{Given $\mathcal A^+$, RCA identifies a root set
$R\in\mathcal R_0$ that jointly explains the active alarms.}

\subsection{Overview}

MATERO-RCA evaluates a root-set hypothesis by minimizing alarm and
local-compatibility energies over counterfactual trajectories permitted by its
root-effect modes. For $R$ and $\mathbf m$, let $\mathcal F(R,\mathbf m)$
denote this feasible set. Here $\Phi_A$ penalizes unresolved active alarms,
$\Phi_C$ penalizes departures from normal local causal relations, and
$\lambda\geq0$ balances them. The trajectory objective and its constrained
inner value are
\begingroup
\small
\begin{equation}
J(\widetilde{\mathbf x})=
\Phi_A(\widetilde{\mathbf x})+\lambda\Phi_C(\widetilde{\mathbf x}),\;
\mathcal E(R,\mathbf m)=
\min_{\widetilde{\mathbf x}\in\mathcal F(R,\mathbf m)}J(\widetilde{\mathbf x}).
\label{eq:counterfactual-energy}
\end{equation}
\endgroup
Thus, $\mathcal E(R,\mathbf m)$ measures how well a root--mode assignment can
return the event to model-defined normal operation.

With cardinality weight $\gamma\geq0$, the outer objective and ideal
root--mode assignment are
\begingroup
\small
\begin{equation}
U(R,\mathbf m)=\mathcal E(R,\mathbf m)+\gamma|R|,\;
(R^\star,\mathbf m^\star)\in
\operatorname*{arg\,min}_{R\in\mathcal R_0,\mathbf m}U.
\label{eq:root-set-optimization}
\end{equation}
\endgroup
We assume the generating root set is irreducible: every lower-cardinality
alternative increases the inner value more than it reduces the cardinality
penalty. CompatNet learns temporal likelihoods over local causal relations and
alarm-parent contexts, then calibrates their negative log-likelihoods to define
$\Phi_C$ and $\Phi_A$.
The mixed-type inner problem is nonconvex. RepairNet initializes mode-aware
counterfactual trajectories after self-supervised training on pseudo-corrupted
normal windows. Gradient refinement then optimizes them under
Eq.~\eqref{eq:counterfactual-energy}. Evaluating a root-set hypothesis requires
costly inner-oracle calls across admissible mode assignments, making the outer
problem combinatorial. Residual-cover best-bound search prioritizes the most
competitive unseen root-set hypothesis while certifying fixed-oracle outer
optimality. Figure~\ref{fig:framework} summarizes the inference framework.

\subsection{Graph-Factored Temporal Compatibility}

The variable-level graph $G$ specifies local causal relations but not trajectory
dynamics. Dynamic-SCM rollouts \citep{weilbach2024counterfactual} and
Granger-causal exogenous-variable models \citep{han2025root} require explicit
generative models, which are difficult to identify with stochastic responses or
unobserved high-level commands. Rather than commit to a single forward
realization, CompatNet learns distributional support over admissible temporal
responses. Under assigned root-effect modes, a counterfactual trajectory is
model-valid when its alarm-parent contexts and relevant local causal relations
fall within calibrated normal support. CompatNet maps this criterion to
differentiable energies for Eq.~\eqref{eq:counterfactual-energy} through a
local-relation branch $p_\theta$ and an alarm-context branch $q_\psi$, discussed as follows.

\paragraph{State-Conditioned Finite-Horizon Likelihood.}
Industrial responses may depend on their state at the start of a local horizon.
To retain this dependence without a long autoregressive rollout, CompatNet
scores overlapping $H_i$-step segments conditioned on their anchor states. For
a child $X_i$, let $\mathcal B_i$ index anchors covering the event window:
\begin{equation}
\begin{aligned}
e_i^C(\mathbf x)
&=-\frac{1}{|\mathcal B_i|H_i}\sum_{a\in\mathcal B_i}
\log p_{\theta}\!\left(\mathbf x_{i,a+1:a+H_i}\right.\\
&\hspace{17mm}\left.\mid x_{i,a},
\mathbf x_{\operatorname{pa}(X_i),a:a+H_i}\right).
\end{aligned}
\label{eq:compatnet-local-energies}
\end{equation}
Here $x_{i,a}$ is the anchor response state. A gated recurrent unit (GRU)
\citep{cho2014learning} initializes from the child--parent anchor
and processes relation-specific embeddings of future parent states to score
the child segment. For responses independent of their initial states,
CompatNet uses a full-window likelihood without rollout. Details
are provided in the supplementary material.

\paragraph{Categorical Joint-Parent State Likelihood.}
Normal windows contain no active alarms, so learning
$p(A_j\mid\operatorname{pa}(A_j))$ degenerates to the inactive label. Yet the
parents may be individually plausible while their joint temporal configuration
is not. CompatNet embeds the parent trajectories and predicts their categorical
joint state $s_{j,t}$, discretizing continuous parents when needed. Because
plausibility depends on temporal order, the state sequence is factorized
autoregressively:
\begin{equation}
e_j^A(\mathbf x)
=-\frac{1}{T}\sum_{t=0}^{T-1}
\log q_{\psi}(s_{j,t}\mid s_{j,0:t-1}).
\label{eq:compatnet-alarm-energy}
\end{equation}
This energy measures normal support for the joint alarm-parent context rather than
predicting $A_j$. Accordingly, an active alarm is resolved when its context
returns to learned normal support.

\paragraph{Likelihood Learning.}
Both branches maximize normal-window likelihoods, optionally regularized by
relation-breaking negatives. Training details are provided in the
supplementary material.

\paragraph{Calibrated Compatibility Energies.}
Raw likelihood energies have relation-specific normal baselines. CompatNet
estimates high empirical quantiles $\tau_i^C$ and $\tau_j^A$ from held-out
normal energies \citep{umsonst2023finite}. For an alarmed event, let
$\mathcal I=\{X_i\in\mathcal X:\operatorname{pa}(X_i)\ne\varnothing\}
\cup\mathcal A^+$ index graph nodes with energy terms. CompatNet defines
\begin{equation}
\begin{aligned}
\Phi_C(\mathbf x)&=
\sum_{X_i\in\mathcal I\cap\mathcal X}\phi_i^C(\mathbf x)
=\sum_{X_i\in\mathcal I\cap\mathcal X}[e_i^C(\mathbf x)-\tau_i^C]_+,\\
\Phi_A(\mathbf x)&=
\sum_{A_j\in\mathcal I\cap\mathcal A^+}\phi_j^A(\mathbf x)
=\sum_{A_j\in\mathcal I\cap\mathcal A^+}[e_j^A(\mathbf x)-\tau_j^A]_+.
\end{aligned}
\label{eq:compatnet-calibration}
\end{equation}
Here $[z]_+=\max(z,0)$, $\phi_i^C$ is the calibrated energy for the local
relation $\operatorname{pa}(X_i)\!\to\!X_i$, and $\phi_j^A$ for the alarm-parent
context $\operatorname{pa}(A_j)$. During normal-window
training, alarm-context terms are formed for all $A_j\in\mathcal A$.
Calibration assigns zero penalty within estimated normal support, admitting
multiple normal trajectories rather than a single rollout.

\subsection{Mode-Aware Counterfactual Trajectory Optimization}

For fixed $(R,\mathbf m)$, minimizing $J$ over $\mathcal F(R,\mathbf m)$ is a
nonconvex mixed-domain problem sensitive to initialization. RepairNet supplies
diverse initial trajectories for subsequent optimization under $J$.

\paragraph{Mode-Aware Counterfactual Trajectories.}
For root $r$, let $\mathcal M_{r,\mathrm{o}}=\{r\}$. Under mode
$\mathrm{p}$, $\mathcal M_{r,\mathrm{p}}$ additionally includes descendants
reached through observed local causal relations satisfying
$\phi_i^C(\mathbf x)\leq\varepsilon$. The mutable scope is
$\mathcal M(R,\mathbf m)=\bigcup_{r\in R}\mathcal M_{r,m_r}$.
Then $\mathcal F(R,\mathbf m)$ fixes
$\widetilde{\mathbf x}_i=\mathbf x_i$ for every
$X_i\notin\mathcal M(R,\mathbf m)$.

\paragraph{Local Multi-Proposal Learning.}
RepairNet is self-supervised on pseudo-corrupted normal windows without repair
labels. For each child $X_i$, it perturbs $\mathbf x_i$ into $\mathbf x_i'$
while fixing its causal-parent trajectories, then reconstructs $\mathbf x_i$. Following
stochastic multiple-choice learning and its adaptation to time-series forecasting
\citep{lee2016stochastic,cortes2025winner}, a shared encoder with relation-specific
decoders generates $L$ candidates:
\begin{equation}
\begin{aligned}
\mathbf h_i&=\operatorname{Enc}_{\omega}
(\mathbf x_i',\mathbf x_{\operatorname{pa}(X_i)},i),\quad
\boldsymbol\eta_i^{(1)}=\operatorname{Dec}_{\omega,i}(\mathbf h_i),\\
\boldsymbol\eta_i^{(\ell)}
&=\boldsymbol\eta_i^{(1)}+
\Delta_{\omega,i}(\mathbf h_i,\ell,\mathbf z_\ell),\quad
\widetilde{\mathbf x}_i^{(\ell)}
=g_i(\boldsymbol\eta_i^{(\ell)}).
\end{aligned}
\label{eq:repairnet-proposals}
\end{equation}
The encoder is shared across relations, but decoders and residual heads are
relation-specific. The first proposal uses deterministic decoder logits. For
$\ell>1$, the residual head conditions on the proposal index and Gaussian noise
$\mathbf z_\ell\sim\mathcal N(0,I)$ to generate stochastic alternatives. The
type-specific map $g_i$ gives a differentiable relaxed trajectory.

Replacing $\mathbf x_i$ in the clean window with each candidate from
Eq.~\eqref{eq:repairnet-proposals} yields the full proposal set
$\{\widetilde{\mathbf x}^{(\ell)}\}_{\ell=1}^{L}$. Training combines reconstruction,
low energy, and proposal diversity:
\begin{equation}
\min_{\omega}\quad
\mathcal L_r(\widetilde{\mathbf x}^{(1)},\mathbf x)
+\alpha \mathcal L_e(\{\widetilde{\mathbf x}^{(\ell)}\})
+\beta \mathcal L_d(\{\widetilde{\mathbf x}^{(\ell)}\}).
\label{eq:repairnet-training}
\end{equation}
Here $\mathcal L_r$ compares the deterministic proposal with its
clean target using the corresponding domain loss.
$\mathcal L_e$ is a Gibbs-weighted risk over trajectory objectives,
and $\mathcal L_d$ penalizes insufficient pairwise diversity
\citep{perera2024annealed}. Network and training details are provided in the
supplementary material. At
inference, these learned proposals initialize the mutable variables in
$\mathcal M(R,\mathbf m)$.

\paragraph{Objective-Directed Refinement.}
Because RepairNet is trained on pseudo-corrupted normal windows, its test-event
proposals may remain suboptimal under $J$. We therefore initialize gradient
refinement from them. Binary and categorical variables retain their native
supports, while continuous variables use ordered bins fixed from normal data.
Direct hard selection over these
supports yields zero gradients almost everywhere. For proposal $\ell$,
$\boldsymbol\xi_0^{(\ell)}$ is initialized from its assembled proposal logits,
and $\boldsymbol\xi_s^{(\ell)}$ denotes the refinement logits at step $s$. We
apply an annealed straight-through (ST) estimator
\citep{bengio2013estimating,jang2017categorical} and update
\begin{equation}
\begin{gathered}
\widetilde{\mathbf x}_s^{(\ell)}
=\operatorname{ST}_{\tau_s}(\boldsymbol\xi_s^{(\ell)}),\\
\boldsymbol\xi_{s+1}^{(\ell)}
=\boldsymbol\xi_s^{(\ell)}-\eta_s\nabla_{\boldsymbol\xi}
J(\widetilde{\mathbf x}_s^{(\ell)}).
\end{gathered}
\label{eq:trajectory-refinement}
\end{equation}
Here $\operatorname{ST}_{\tau_s}$ selects the highest-logit valid value in the forward pass and
uses a temperature-scaled sigmoid for binary values or softmax for other
discrete values in the backward pass. The temperature $\tau_s$ is annealed and
$\eta_s$ is the step size. Only trajectories in $\mathcal M(R,\mathbf m)$ are
updated. Each forward trajectory is therefore valid and evaluated by $J$. The
lowest trajectory objective over all proposals and refinement steps defines the approximate
inner value $\widehat{\mathcal E}(R,\mathbf m)=
\min_{\ell,s}J(\widetilde{\mathbf x}_s^{(\ell)})$, which approximates
$\mathcal E(R,\mathbf m)$ in Eq.~\eqref{eq:counterfactual-energy}.

\subsection{Certified Residual-Cover Best-Bound Search}

Enumerating Eq.~\eqref{eq:root-set-optimization} evaluates exponentially many
root--mode assignments through the costly inner solver. The
objective's nonnegative additive energy yields a residual-cover lower bound.
An exact MILP minimizes this bound over unseen root sets, selecting the most
promising set and stopping once none can improve the best value. For a fixed
deterministic inner solver, let $\widehat{\mathbf x}_{R,\mathbf m}$ denote its
output for $(R,\mathbf m)$. Using
$\widehat{\mathcal E}(R,\mathbf m)=J(\widehat{\mathbf x}_{R,\mathbf m})$, define
$\widehat U(R,\mathbf m)=\widehat{\mathcal E}(R,\mathbf m)+\gamma|R|$,
with $\widehat U_R=\min_{\mathbf m}\widehat U(R,\mathbf m)$. The exact
objective in Eq.~\eqref{eq:root-set-optimization} satisfies
$U(R,\mathbf m)\leq\widehat U(R,\mathbf m)$. All guarantees concern this fixed
inner solver.

The lower bound follows from the additive decomposition in
Eq.~\eqref{eq:counterfactual-energy}. Each $v\in\mathcal I$ indexes one
calibrated energy term in Eq.~\eqref{eq:compatnet-calibration}. Let
$S_v\subseteq\mathcal X$ denote the variable trajectories on which that
term depends. Specifically, $S_v=\{v\}\cup\operatorname{pa}(v)$ for
$v\in\mathcal X$ and $S_v=\operatorname{pa}(v)$ for $v\in\mathcal A^+$.
Its contribution to $J(\mathbf x)$ is $c_v=\lambda\phi_i^C(\mathbf x)$ for
$v=X_i$ and $c_v=\phi_j^A(\mathbf x)$ for $v=A_j$. The residual-cover lower
bound is
\begin{equation}
L(R,\mathbf m)=\gamma|R|+
\sum_{v\in\mathcal I}c_v\,
\mathbf 1[S_v\cap\mathcal M(R,\mathbf m)=\varnothing].
\label{eq:residual-cover-bound}
\end{equation}
The indicator retains $c_v$ when the root--mode assignment cannot change any trajectory
in $S_v$. Otherwise, that energy is optimistically set to zero.

\begin{proposition}
If calibrated energies are nonnegative and the fixed inner solver returns a
trajectory in $\mathcal F(R,\mathbf m)$, then
$L(R,\mathbf m)\leq U(R,\mathbf m)\leq
\widehat U(R,\mathbf m)$.
\end{proposition}

\paragraph{Best-Bound Search.}
Let $\widehat U_0=+\infty$. At iteration $t$, let
$\mathcal R_t=\mathcal R_0\setminus\{R_s:s<t\}$ contain the unseen root sets
and $\widehat U_t=\min_{s<t}\widehat U_{R_s}$ be the best evaluated objective
value.
\emph{(1) Selection.} Select the unseen root set with the smallest lower bound:
\begingroup
\small
\begin{equation}
R_t\in\arg\min_{R\in\mathcal R_t}\min_{\mathbf m}L(R,\mathbf m),
\qquad
B_t=\min_{\mathbf m}L(R_t,\mathbf m).
\label{eq:best-bound-selection}
\end{equation}
\endgroup
Set $B_t=+\infty$ if $\mathcal R_t=\varnothing$.
\emph{(2) Certification.} For tolerance $\epsilon\geq0$, terminate if
$B_t>\widehat U_t+\epsilon$.
\emph{(3) Evaluation.} Otherwise, evaluate all admissible modes of $R_t$ and
set $\widehat U_{R_t}=\min_{\mathbf m}\widehat U(R_t,\mathbf m)$.
\emph{(4) Update.} Set
$\widehat U_{t+1}=\min\{\widehat U_t,\widehat U_{R_t}\}$ and
$\mathcal R_{t+1}=\mathcal R_t\setminus\{R_t\}$, then repeat.
Algorithm~\ref{alg:matero-inference} summarizes the procedure.
Proposition~2 shows that Eq.~\eqref{eq:best-bound-selection} can be solved
exactly, while Theorem~1 establishes optimality certificates for the resulting
best-bound search.

\begin{algorithm}[tb]
\caption{MATERO-RCA Best-Bound Search}
\label{alg:matero-inference}
\textbf{Input}: Alarmed temporal event $(\mathbf x,\mathcal A^+)$, graph $G$, frozen
CompatNet and RepairNet, $\mathcal R_0$, and $\epsilon\geq0$\\
\textbf{Output}: $(\widehat R^\star,\widehat{\mathbf m}^\star,
\widehat{\mathbf x}^{\star},\widehat U_t)$
\begin{algorithmic}[1]
\State Compute Eq.~\eqref{eq:residual-cover-bound} and initialize
  $t\leftarrow0$ and $\widehat U_0\leftarrow+\infty$
\While{$\mathcal R_t\ne\varnothing$}
  \State Solve the MILP Eq.~\eqref{eq:best-bound-selection} to
    obtain $R_t$ and $B_t$
  \State \textbf{if} $B_t>\widehat U_t+\epsilon$ \textbf{then break}
  \ForAll{admissible $\mathbf m\in\{\mathrm{o},\mathrm{p}\}^{|R_t|}$}
\State Generate $\{\widetilde{\mathbf x}^{(\ell)}\}_{\ell=1}^{L}$ by
      Eq.~\eqref{eq:repairnet-proposals}
    \State Refine all proposals by Eq.~\eqref{eq:trajectory-refinement} over
      $\mathcal F(R_t,\mathbf m)$ to obtain
      $\widehat{\mathbf x}_{R_t,\mathbf m}$
    \State $\widehat U(R_t,\mathbf m)\leftarrow
      J(\widehat{\mathbf x}_{R_t,\mathbf m})+\gamma|R_t|$
  \EndFor
  \State $\mathbf m_t\leftarrow\arg\min_{\mathbf m}\widehat U(R_t,\mathbf m)$
    and $\widehat U_{R_t}\leftarrow\widehat U(R_t,\mathbf m_t)$
  \State If $\widehat U_{R_t}<\widehat U_t$, set
    $(\widehat R^\star,\widehat{\mathbf m}^\star,\widehat{\mathbf x}^{\star})
    \leftarrow(R_t,\mathbf m_t,\widehat{\mathbf x}_{R_t,\mathbf m_t})$
  \State $\widehat U_{t+1}\leftarrow\min\{\widehat U_t,\widehat U_{R_t}\}$,
    $\mathcal R_{t+1}\leftarrow\mathcal R_t\setminus\{R_t\}$, and
    $t\leftarrow t+1$
\EndWhile
\State \Return $(\widehat R^\star,\widehat{\mathbf m}^\star,
  \widehat{\mathbf x}^{\star},\widehat U_t)$
\end{algorithmic}
\end{algorithm}

\begin{proposition}[Exact MILP encoding]
Equation~\eqref{eq:best-bound-selection} has an exact binary MILP encoding with
auxiliary variables. Its selected root set is $R_t$, and its optimal value is
$B_t$.
\end{proposition}
The full encoding and proof are provided in the supplementary material, and we
solve the MILP with HiGHS \citep{huangfu2018parallelizing}.

\begin{theorem}[Fixed-oracle best-bound certificate]
For the fixed inner solver, let
$\widehat U^\star=\min_{R\in\mathcal R_0}\widehat U_R$. At iteration $t$, after
solving Eq.~\eqref{eq:best-bound-selection} and before evaluating $R_t$,
Propositions~1--2 imply
\begin{equation}
\begin{gathered}
B_t\leq\min_{R\in\mathcal R_t}\widehat U_R,\\
0\leq\widehat U_t-\widehat U^\star
\leq[\widehat U_t-B_t]_+ .
\end{gathered}
\label{eq:anytime-gap}
\end{equation}
Thus $B_t\geq\widehat U_t$ certifies fixed-oracle outer optimality, while
$B_t>\widehat U_t+\epsilon$ certifies that every root set with
$\widehat U_R\leq\widehat U^\star+\epsilon$ has been evaluated. The search
performs at most $|\mathcal R_0|$ completed root-set evaluations.
\end{theorem}

\begin{corollary}[Certified root-set separation]
Let $\widehat R^\star$ be an evaluated root set attaining the fixed-oracle
optimum certified by Theorem~1, and let $\widehat U_t^{\mathrm{alt}}$ be the
smallest objective value among evaluated root sets distinct from
$\widehat R^\star$.
Define
\begin{equation*}
\underline\Delta_t=\min\{\widehat U_t^{\mathrm{alt}},B_t\}-\widehat U_t,
\end{equation*}
where $\widehat U_t^{\mathrm{alt}}=+\infty$ if no such set exists. Then
$0\leq\underline\Delta_t\leq\widehat U_R-\widehat U^\star$ for every
$R\ne\widehat R^\star$. If $\underline\Delta_t>\epsilon$,
$\widehat R^\star$ is the unique $\epsilon$-optimal root set.
\end{corollary}
Proofs are given in the supplementary material. These guarantees concern the
finite admissible root--mode space under the fixed inner solver, rather than
global optimization over continuous trajectories.

\section{Experiments}

\begin{table*}[t]
\centering
{\small
\setlength{\tabcolsep}{0.55mm}
\renewcommand{\arraystretch}{0.96}
\begin{tabular}{@{}l*{6}{ccc}@{}}
\toprule
Method
& \multicolumn{3}{c}{causRCA}
& \multicolumn{3}{c}{TA}
& \multicolumn{3}{c}{TS}
& \multicolumn{3}{c}{LM}
& \multicolumn{3}{c}{QT}
& \multicolumn{3}{c}{Mean} \\
\cmidrule(lr){2-4}\cmidrule(lr){5-7}\cmidrule(lr){8-10}
\cmidrule(lr){11-13}\cmidrule(lr){14-16}\cmidrule(lr){17-19}
& A@1 & C@3 & Set F1
& A@1 & C@3 & Set F1
& A@1 & C@3 & Set F1
& A@1 & C@3 & Set F1
& A@1 & C@3 & Set F1
& A@1 & C@3 & Set F1 \\
\midrule
EasyRCA$^\dagger$
& 88.2 & 80.3 & 81.0
& 0.0 & 13.3 & 0.0
& 60.0 & \underline{83.3} & 58.9
& 14.0 & 4.0 & 14.3
& 57.1 & 77.1 & 53.9
& 43.9 & 51.6 & 41.6 \\
T-RCA$^\dagger$
& 88.2 & 80.3 & 81.0
& 0.0 & 13.3 & 0.0
& 56.7 & \underline{83.3} & 55.6
& 14.0 & 4.0 & 14.3
& \underline{97.1} & 94.3 & \underline{94.8}
& 51.2 & 55.0 & 49.1 \\
AERCA
& 68.3 & 79.3 & 71.4
& 10.0 & 26.7 & 8.9
& 53.3 & 63.3 & 46.7
& 0.0 & 0.0 & 0.0
& 45.7 & 74.3 & 57.1
& 35.5 & 48.7 & 36.8 \\
AERCA$^\dagger$
& 74.4 & 77.8 & 70.2
& 3.3 & 26.7 & 5.6
& 40.0 & 56.7 & 36.7
& 2.0 & 2.0 & 2.0
& 60.0 & 74.3 & 60.4
& 36.0 & 47.5 & 35.0 \\
CIRCA
& 78.3 & \textbf{100.0} & 89.5
& 43.3 & 70.0 & 47.2
& 50.0 & \underline{83.3} & 49.4
& 30.0 & 32.0 & 33.3
& \underline{97.1} & \underline{97.1} & 66.5
& 59.8 & 76.5 & 57.2 \\
RCG
& \underline{93.6} & \textbf{100.0} & 88.7
& \underline{83.3} & \textbf{100.0} & \underline{86.7}
& \underline{93.3} & \textbf{100.0} & \underline{91.6}
& \underline{96.0} & \underline{92.0} & 94.0
& 94.3 & 88.6 & 90.5
& \underline{92.1} & \underline{96.1} & \underline{90.3} \\
StableRCA
& 92.4 & \underline{92.3} & \underline{90.4}
& 23.3 & 83.3 & 36.0
& \textbf{96.7} & \textbf{100.0} & 89.4
& \textbf{100.0} & 82.0 & \underline{97.3}
& 88.6 & \textbf{100.0} & 90.3
& 80.2 & 91.5 & 80.7 \\
SmoothTraversal
& 67.5 & 85.2 & 62.9
& 0.0 & 50.0 & 8.9
& 63.3 & 80.0 & 62.2
& 36.0 & 22.0 & 39.0
& 28.6 & 71.4 & 32.4
& 39.1 & 61.7 & 41.1 \\
BARO
& 87.2 & 88.4 & 85.5
& 0.0 & 13.3 & 0.0
& 40.0 & 80.0 & 40.0
& 0.0 & 0.0 & 0.0
& 68.6 & 94.3 & 90.5
& 39.2 & 55.2 & 43.2 \\
IDI$^\dagger$
& 80.4 & 89.9 & 79.3
& 66.7 & \underline{93.3} & 58.3
& 30.0 & 73.3 & 50.0
& \underline{96.0} & 86.0 & 91.0
& 82.9 & 85.7 & 78.4
& 71.2 & 85.7 & 71.4 \\
\midrule
MATERO-RCA
& \textbf{100.0} & \textbf{100.0} & \textbf{100.0}
& \textbf{96.7} & \textbf{100.0} & \textbf{96.7}
& 83.3 & \textbf{100.0} & \textbf{94.9}
& \textbf{100.0} & \textbf{100.0} & \textbf{100.0}
& \textbf{100.0} & \textbf{100.0} & \textbf{97.5}
& \textbf{96.0} & \textbf{100.0} & \textbf{97.8} \\
\bottomrule
\end{tabular}
}
\caption{Main results (\%). A@1, C@3, and Set F1 denote AnyRoot@1,
CompleteRoots@3, and set-level F1, respectively.}
\label{tab:main-results}
\end{table*}

\begin{table}[t]
\centering
{\footnotesize
\setlength{\tabcolsep}{0.20mm}
\begin{tabular}{@{}lcc@{\hspace{0.4mm}}lcc@{\hspace{0.4mm}}lcc@{}}
\toprule
Method & ES & O-ES$^*$ & Method & ES & O-ES$^*$ & Method & ES & O-ES$^*$ \\
\cmidrule(lr){1-3}\cmidrule(lr){4-6}\cmidrule(lr){7-9}
EasyRCA$^\dagger$ & 33.6 & -- &
CIRCA & 30.5 & 54.1 &
BARO & 41.0 & 38.0 \\
T-RCA$^\dagger$ & 44.3 & -- &
RCG & \underline{81.0} & \textbf{85.9} &
IDI$^\dagger$ & 44.2 & 64.3 \\
AERCA & 29.1 & 32.8 &
Stable. & 70.9 & \underline{69.7} &
MATERO. & \textbf{94.3} & -- \\
AERCA$^\dagger$ & 29.2 & 31.8 &
Smooth. & 32.5 & 33.7 &
& & \\
\bottomrule
\end{tabular}
}
\caption{Macro-averaged projected ES and oracle-cardinality O-ES$^*$ across
five dataset groups (\%).}
\label{tab:exact-set-results}
\end{table}

\subsection{Experimental Setup}
\paragraph{Datasets.}
We evaluate on one real-world industrial dataset,
\textbf{causRCA}~\citep{mehling2026enabling},
three controlled synthetic datasets---\textbf{Temporal Actuator (TA)},
\textbf{Typed Source (TS)}, and the scalability-oriented
\textbf{Large Modular (LM)}---and one
physics-based simulation, \textbf{Quadruple Tank (QT)}. causRCA combines real
normal-operation recordings with labeled hardware-in-the-loop anomaly events and
provides \emph{Probe}, \emph{Coolant}, \emph{Hydraulics}, and \emph{Full}
settings with 11, 15, 17, and 92 nodes (including variables and alarms), respectively. Their test sets contain
34, 25, 41, and 100 events. TA, TS, and LM contain 10, 10, and 206 nodes
and 30, 30, and 50 test events, respectively. QT contains 24 mixed-type nodes
and 35 events, generated from the
published nonlinear quadruple-tank equations~\citep{johansson2000quadruple}.
The provided causal graphs are expert-derived for causRCA, generator-defined for
TA/TS/LM, and physics-derived for QT.

\paragraph{Baselines and Protocols.}
Temporal RCA baselines include EasyRCA$^\dagger$~\citep{assaad2023root},
T-RCA$^\dagger$~\citep{zan2024fly}, and AERCA with its graph-constrained
variant~\citep{han2025root}. General causal and score-based baselines include
CIRCA~\citep{li2022causal}, RCG~\citep{ikram2025root},
StableRCA~\citep{lin2026stablerca}, SmoothTraversal~\citep{orchard2025root},
and BARO~\citep{pham2024baro}. We further include
IDI$^\dagger$~\citep{nagalapatti2025robust} as the closest implemented
set-aware interventional baseline. The $\dagger$ marks method-level variants:
IDI$^\dagger$ replaces the constant binary-alarm target with temporal
alarm-parent contexts, while AERCA$^\dagger$ constrains its learned
Granger-causal coefficients to the provided graph. EasyRCA$^\dagger$ and
T-RCA$^\dagger$ use our adapter to estimate their required anomaly onsets from
event trajectories using normal-only calibration.

Root-set construction considers only ancestors of top-level active alarms.
Ranking methods form sets using Holm's step-down correction~\citep{holm1979simple}
on candidate $p$-values from held-out normal windows ($\alpha=0.01$), while
MATERO-RCA, EasyRCA$^\dagger$, and T-RCA$^\dagger$ retain their set outputs.
Other $\alpha$ values and Benjamini--Hochberg selection~\citep{benjamini1995controlling}
are evaluated in the supplementary material.
Sets are projected as needed by adding the highest-ranked eligible ancestors
until all top-level active alarms are covered. To compute A@1 and C@3 from
MATERO-RCA's $\widehat R^\star$, we rank its members before other variables
and order each group by
$\min_{m\in\{\mathrm{o},\mathrm{p}\}}L(\{r\},m)$. This ranking uses neither
root labels nor the true root count.

For evaluation metrics, we report
\textbf{AnyRoot@1} (top-1 hits a root),
\textbf{CompleteRoots@3} (top-3 covers all roots),
\textbf{Set F1} (mean per-event harmonic mean of reported-set precision and
recall), and
\textbf{ExactSet} (ES; the reported set exactly matches the annotated roots).
For ranking methods, \textbf{oracle-cardinality ES} (O-ES$^*$) forms the set using the
annotated root count and serves only as a top-$k$ comparator unavailable at
test time.
Additional metrics are included in the supplementary material.

\paragraph{Implementation.}
We set $K=3$, $\lambda=0.5$, $\gamma=0.25$, and
$\varepsilon=0.01$, using the $0.99$ and $0.95$ quantiles to
calibrate $e_i^C$ and $e_j^A$, respectively. Eq.~\eqref{eq:trajectory-refinement}
iterates for 50 refinement steps with the AdamW optimizer~\citep{loshchilov2019decoupled}.
All experiments and runtime measurements use an Intel Core i7-12700F CPU and
one NVIDIA RTX 3080 GPU; further implementation details are provided
in the supplementary material.

\afterpage{%
  \begin{figure}[t]
    \centering
    \includegraphics[width=0.9\columnwidth]{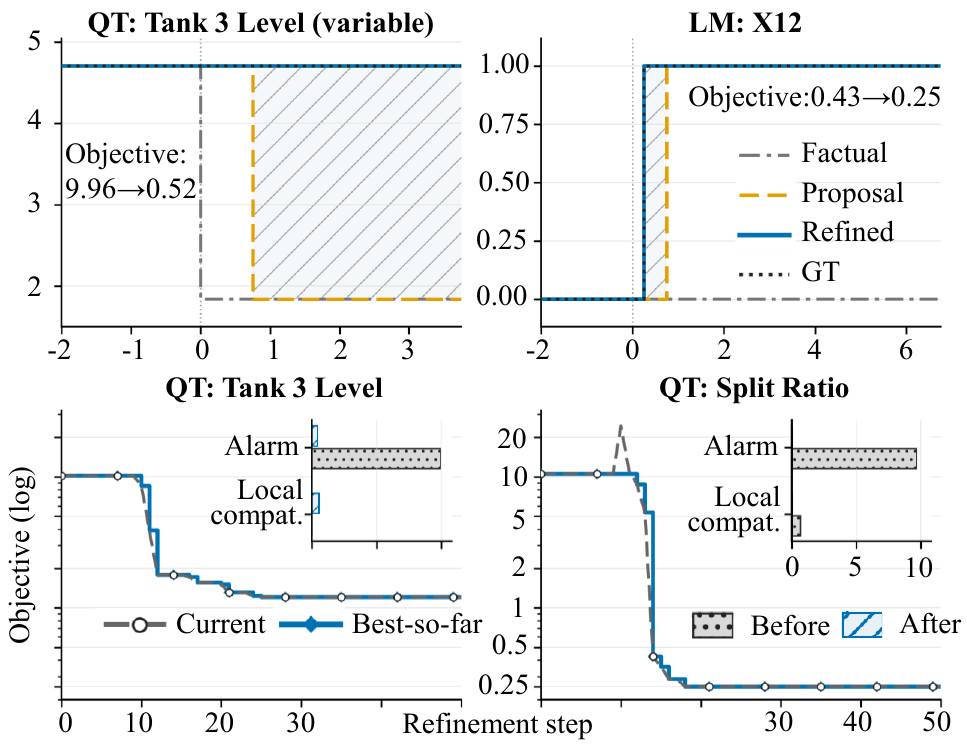}
    \caption{Initial proposal vs.\ refinement. Refinement adjusts trajectories
    and lowers the objective.}
    \label{fig:optimization-behavior-ab}
  \end{figure}
}

\subsection{Main Results}

Tables~\ref{tab:main-results} and~\ref{tab:exact-set-results} report root
ranking, set recovery, and exact-set identification. MATERO-RCA performs
consistently across five dataset groups, achieving 97.8\% mean Set F1 and
94.3\% ES. The gap reflects two challenges.
First, existing methods provide incomplete root evidence. 
RCG, StableRCA, CIRCA, and AERCA rank variables by conditional shifts,
residuals, or learned dynamics, potentially confusing roots with causal
neighbors when observation-only effects appear as relation violations rather
than marginal anomalies. AERCA and IDI rely on fitted temporal or structural mechanisms.
Their lower performance on TA and QT suggests difficulty under stochastic delays and coupled temporal responses,
while BARO and SmoothTraversal use marginal anomaly evidence. 
Even with our onset adapter, EasyRCA$^\dagger$ and T-RCA$^\dagger$ depend on
anomalous-node preidentification, which is difficult for contextual anomalies.
Second, ranking methods do not infer root-set cardinality. Consequently, their
native ES falls markedly below the oracle-cardinality top-$k$ comparator
O-ES$^*$. MATERO-RCA addresses both by jointly inferring root-effect modes and
root sets using learned graph-factored temporal compatibility, yielding the
best performance. By best-baseline Set F1,
Probe is the most challenging causRCA subsystem: the strongest baseline reaches
79.6\%, whereas MATERO-RCA achieves 100.0\%. Per-setting results and additional
metrics are in the supplementary material.

\subsection{Analysis}

\begin{table}[b]
  \centering
  \footnotesize
  \setlength{\tabcolsep}{3pt}
  \renewcommand{\arraystretch}{0.94}
  \begin{tabular*}{\columnwidth}{@{\extracolsep{\fill}}lccc@{}}
    \toprule
    Variant & Set F1 $\uparrow$ & ES $\uparrow$ & Root--mode evals. $\downarrow$ \\
    \midrule
    Full & \textbf{97.8} & \textbf{94.3} & 4.4 \\
    w/o alarm term $\Phi_A$ & 85.1 & 79.1 & 4.8 \\
    w/o compatibility term $\Phi_C$ & 85.7 & 83.0 & 9.3 \\
    w/o gradient refinement & 92.9 & 84.3 & 19.0 \\
    w/o mode $\mathrm{p}$ & 87.5 & 67.8 & 3.6 \\
    w/o mode $\mathrm{o}$ & 94.1 & 89.6 & \textbf{1.4} \\
    \bottomrule
  \end{tabular*}
  \caption{Five-group macro ablation study.}
  \label{tab:component-ablation}
\end{table}

\noindent\textbf{Ablation studies.} Table~\ref{tab:component-ablation}
isolates the objective terms, refinement, and root-effect modes. A root--mode
evaluation is one complete proposal-and-refinement solve for $(R,\mathbf m)$.
Removing either $\Phi_A$ or $\Phi_C$ reduces set recovery, confirming that
alarm resolution and temporal compatibility provide complementary evidence.
Gradient refinement lowers evaluated objective values and tightens $\widehat U_t$,
allowing $B_t>\widehat U_t+\epsilon$ to hold earlier and reducing root--mode
evaluations. Removing mode $\mathrm{p}$ or $\mathrm{o}$ reduces root--mode
evaluations from 4.4 to 3.6 or 1.4, respectively, because fewer mode
assignments are admissible. However, this restriction omits valid root-effect
modes, lowering both Set F1 and ES.

\noindent\textbf{Initial proposal vs.\ refinement.}
RepairNet amortizes initial trajectories from normal data, but does
not optimize them for an event's alarm-context and local-relation
energies. Gradient refinement makes proposals event-directed by minimizing
Eq.~\eqref{eq:counterfactual-energy}. Figure~\ref{fig:optimization-behavior-ab}
shows the resulting trajectory adjustments, which reduce the objective from
10.17 to 1.21 and from 10.57 to 0.25 in two representative events.

\noindent\textbf{Exhaustive vs.\ best-bound search.}
Best-bound search prioritizes the unseen root set with the smallest lower
bound, focusing inner solves on the most competitive root-set hypotheses.
Figure~\ref{fig:search-efficiency} shows median reductions of
$4.0{\times}$--$32.5{\times}$ in root--mode evaluations across four test
groups and a median $20.7{\times}$ runtime speedup on aligned QT events.
For LM, exhaustive search would require $8{,}322$
root--mode evaluations per event versus $4.8$ under best-bound search, a
$1{,}734{\times}$ reduction.

\begin{figure}[t]
  \centering
  \includegraphics[width=0.9\columnwidth]{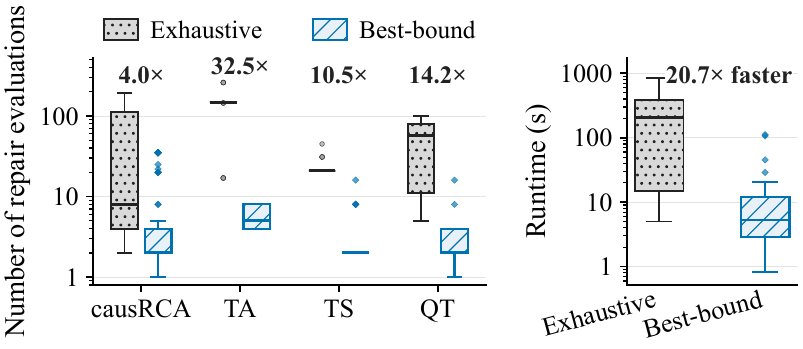}
  \caption{Exhaustive vs.\ best-bound search. Best-bound search reduces root--mode
  evaluations and runtime.}
  \label{fig:search-efficiency}
\end{figure}

\begin{figure}[t]
  \centering
  \includegraphics[width=0.9\columnwidth]{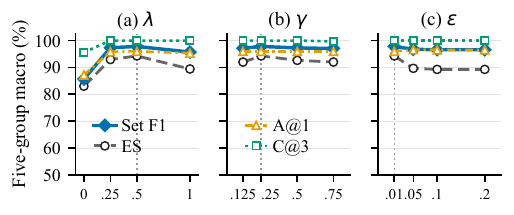}
  \caption{Sensitivity to $\lambda$, $\gamma$, and
  $\varepsilon$. Dotted lines mark defaults; scores are
  five-group macros.}
  \label{fig:objective-sensitivity}
\end{figure}

\noindent\textbf{Deterministic vs.\ stochastic proposals.}
Stochastic proposals in Eq.~\eqref{eq:repairnet-proposals} improve the objective
on only 11 of 345 events (1 causRCA, 3 TA, and 7 QT), tying otherwise, but
raise median runtime from 2.7 to 6.3 seconds. Deterministic initialization
therefore generally suffices, while stochastic proposals remain optional for
multimodal responses.

\noindent\textbf{Hyperparameter sensitivity.}
Figure~\ref{fig:objective-sensitivity} varies $\lambda$, $\gamma$, and
$\varepsilon$ individually. Performance is stable across the tested $\gamma$
values and $\lambda\in[0.25,1]$, whereas $\lambda=0$ degrades both metrics.
In contrast, ES is sensitive to $\varepsilon$, which discretely determines the
descendant scope of physical effects. The conservative default
$\varepsilon=0.01$ admits only near-zero compatibility energies. Larger values
introduce competing explanations and reduce ES more than Set F1.

\noindent\textbf{Additional sensitivities.}
Quantile, graph, and bin-count sensitivities are reported in the supplementary
material.

\FloatBarrier

\noindent\textbf{Failure Analysis.}
MATERO-RCA assumes that physical-propagation effects follow response patterns supported by
normal training data. Seven of its nine ExactSet errors are predefined
out-of-distribution (OOD) events that violate this assumption, which is also
shared by most baselines. Each retains the
annotated roots but adds one causal neighbor, indicating conservative
over-selection beyond learned support. Other cases are detailed in the
supplementary material. 

\section{Conclusion}

We presented MATERO-RCA, which jointly optimizes root sets, auxiliary
root-effect modes, and counterfactual trajectories for industrial RCA.
CompatNet defines graph-factored temporal compatibility, while RepairNet and
gradient refinement construct event-directed trajectories. Residual-cover
best-bound search certifies fixed-oracle outer optimality. Experiments on real
industrial and simulated systems show strong complete-root-set recovery and
substantial search savings. Remaining failures arise from causal-graph
misspecification and physical propagation outside learned normal support,
motivating open-set response models and graph-structure uncertainty.

\clearpage
\makeatletter
\setlength{\@dblfptop}{0pt}
\makeatother

\setcounter{section}{0}
\setcounter{subsection}{0}
\setcounter{equation}{0}
\setcounter{figure}{0}
\setcounter{table}{0}
\setcounter{lemma}{0}
\setcounter{algorithm}{0}
\setcounter{secnumdepth}{2}
\renewcommand{\thesection}{S\arabic{section}}
\renewcommand{\thesubsection}{S\arabic{section}.\arabic{subsection}}
\renewcommand{\theequation}{S\arabic{equation}}
\renewcommand{\thefigure}{S\arabic{figure}}
\renewcommand{\thetable}{S\arabic{table}}
\renewcommand{\thelemma}{S\arabic{lemma}}
\renewcommand{\thealgorithm}{S\arabic{algorithm}}

\section*{Appendix}

\paragraph{Appendix organization.}
This appendix follows the technical and empirical progression of the main
paper. Sections~\ref{app:compatnet}--\ref{app:certified-search} expand the
method, covering temporal compatibility modeling, mode-aware counterfactual
trajectory optimization, network architectures, and certified root-set search.
Section~\ref{app:dataset-details} documents the datasets, causal graphs, event
representation, and synthetic-data generation process, while
Section~\ref{app:temporal-baselines} specifies the temporal baseline protocol.
Section~\ref{app:additional-results} collects the remaining empirical evidence:
ranking diagnostics, setting-level causRCA results, calibration,
discretization, causal-graph sensitivity, and failure analysis.
Throughout, we retain the main-paper notation, refer to its equation numbers
explicitly, and introduce only symbols needed for architecture, training,
generation, and proofs.

\section{Graph-Factored Temporal Compatibility}
\label{app:compatnet}

This section specifies the mixed-type likelihood heads, local likelihood
factorizations, and relation-breaking training objective omitted from the main
paper. The corresponding architectures appear in
Section~\ref{app:architectures}. Every modeled alarm has at least one observed
causal parent. Local-relation and RepairNet modules are defined only for
non-source temporal variables.

\subsection{Mixed-Type Likelihoods}
\label{app:discrete-support}

Industrial event trajectories mix continuous measurements with binary and
categorical states \citep{wang2023hybrid}. Bernoulli and categorical heads
directly model the discrete types; continuous signals require a finite
representation to share the likelihood interface.

\noindent\textbf{Continuous-signal discretization.} Each continuous signal is mapped to
ordered bins fixed from normal training data. This mirrors industrial operating
regions defined by alarm thresholds and deadbands \citep{weng2022evidence}.
Let $C_i$ be the resulting support size and
$y_{i,t}=\mathsf Q_i(x_{i,t})$ the algorithmic state. For discrete variables,
$\mathsf Q_i$ returns the native state index; for continuous variables,
\begin{equation}
\begin{aligned}
\mathsf Q_i(x)=c
&\quad\Longleftrightarrow\quad
b_{i,c}\le x<b_{i,c+1},\\[-1pt]
&\qquad c=0,\ldots,C_i-1,
\end{aligned}
\label{eq:supp-ordered-binning}
\end{equation}
with $-\infty=b_{i,0}<\cdots<b_{i,C_i}=+\infty$. The Temporal Compatibility
Network (CompatNet), Counterfactual Repair Network (RepairNet), and refinement
apply $\mathsf Q_i$ coordinate-wise and therefore share the same discrete
support. Thus $\mathbf y=\mathsf Q(\mathbf x)$ is the internal discrete
representation of $\mathbf x$ in Eq.~(3) of the main paper, and $p_\theta$ below
is its induced mass on this support. Decoding uses interior-bin midpoints; exterior bins extrapolate
half the nearest finite-edge gap, or half a raw unit when only one edge exists.

\noindent\textbf{Likelihood heads.} Given context $\mathbf h_{i,t}^C$, the typed
head outputs a Bernoulli probability $\pi_{\theta,i,t}$ or a categorical
probability vector $\boldsymbol\pi_{\theta,i,t}$. The observed-state likelihood is
\begin{equation}
p_{\theta}(y_{i,t}\mid\mathbf h_{i,t}^C)=
\begin{cases}
\pi_{\theta,i,t}^{y_{i,t}}(1-\pi_{\theta,i,t})^{1-y_{i,t}},
& X_i\text{ binary},\\
\boldsymbol\pi_{\theta,i,t}[y_{i,t}],
& X_i\text{ nonbinary}.
\end{cases}
\label{eq:supp-typed-heads}
\end{equation}
Equation~\eqref{eq:supp-typed-heads} unifies the signal types: its first branch
scores the observed binary state, while its second scores the native category
or occupied continuous bin. Thus
$-\log p_{\theta}(y_{i,t}\mid\mathbf h_{i,t}^C)$ is the typed conditional
negative log-likelihood used by the local-relation energies; continuous
responses are scored by learned bin probability rather than raw-unit distance
from a point prediction.

\subsection{Local-Relation Likelihoods}
\label{app:local-likelihoods}

We model a local causal relation only when
$\operatorname{pa}(X_i)\ne\varnothing$. For relation
$\operatorname{pa}(X_i)\to X_i$, let $\mathcal B_i$ index
$H_i$-step segments covering the event window, and let $\mathsf E_i^C$ and
$\mathsf E_{i,k}^C$
embed the child and parent states. A relation-specific gated recurrent unit
(GRU) \citep{cho2014learning} computes
\begin{equation}
\begin{aligned}
\mathbf h_{i,a,0}^C
&=\operatorname{Init}_{\theta,i}\!\left(
\left[\mathsf E_i^C(y_{i,a});
\bigl(\mathsf E_{i,k}^C(y_{k,a})\bigr)_
{X_k\in\operatorname{pa}(X_i)}\right]
\right),\\
\mathbf h_{i,a,h}^C
&=\operatorname{GRU}_{\theta,i}\!\left(
\mathbf h_{i,a,h-1}^C,
\bigl[\mathsf E_{i,k}^C(y_{k,a+h})\bigr]_
{X_k\in\operatorname{pa}(X_i)}
\right).
\end{aligned}
\label{eq:supp-rollout-state}
\end{equation}
Parent-indexed tuples are concatenated in a fixed graph order.
Future child states supervise
$p_{\theta}(y_{i,a+h}\mid\mathbf h_{i,a,h}^C)$ but are not recurrent inputs.
Thus $y_{i,a}$ is a boundary condition, not a self-loop in $G$, and the segment
likelihood in Eq.~(3) of the main paper factorizes as
\begin{equation}
\begin{aligned}
&p_{\theta}\!\left(\mathbf y_{i,a+1:a+H_i}\mid y_{i,a},
\mathbf y_{\operatorname{pa}(X_i),a:a+H_i}\right)\\
&\qquad=\prod_{h=1}^{H_i}
p_{\theta}(y_{i,a+h}\mid\mathbf h_{i,a,h}^C).
\end{aligned}
\label{eq:supp-rollout-energy}
\end{equation}
Averaging its normalized negative log over $\mathcal B_i$ yields the energy in
Eq.~(3) of the main paper; overlapping anchors cover longer events without a single deterministic
rollout.

The full-window branch instead uses a relation-specific temporal convolutional
network (TCN) $\mathsf K_i^C$ \citep{bai2018empirical}. Dilated residual
convolutions encode the aligned parent-state sequence at each time step and
score $t=0,\ldots,T-1$ without a child boundary condition:
\begin{equation}
e_i^C(\mathbf x)=-\frac1T\sum_{t=0}^{T-1}
\log p_{\theta}(y_{i,t}\mid\mathbf h_{i,t}^C).
\label{eq:supp-window-energy}
\end{equation}
Equation~\eqref{eq:supp-compat-routing} specifies this full-window context.
The TCN is parameterized independently of the GRU; source variables have no
parent-conditioned local causal relation.
Each non-source relation is routed to exactly one local-relation bank:
configured state-conditioned relations use
Eq.~\eqref{eq:supp-rollout-energy}, and the remainder use
Eq.~\eqref{eq:supp-window-energy}. Hence $e_i^C$ denotes one routed energy per
relation and is never double counted.

\subsection{Likelihood Learning}
\label{app:compatnet-training}

CompatNet trains each local-relation or alarm-context energy
$e\in\{e_i^C,e_j^A\}$ on normal windows.

\noindent\textbf{Negative samples.} For a local causal relation,
$\mathbf x^{\mathrm{neg}}$ retains the
parent trajectories and applies a temporal shift, hold, segment swap, or short
pulse to the future child. For an alarm-parent context, a parent subset receives a
large lag, transition hold, or stable-segment corruption, disrupting the joint
temporal state. These relation-breaking samples sharpen compatibility
boundaries rather than emulate physical anomaly propagation
\citep{carmona2022neural}. For either energy, the objective is
\begin{equation}
\mathcal L_c
=\mathbb E_{\mathbf x\sim\mathcal D_N}\!\left[
e(\mathbf x)+\zeta[\chi+e(\mathbf x)-e(\mathbf x^{\mathrm{neg}})]_+\right].
\label{eq:supp-compat-training}
\end{equation}
Here $\chi>0$ is the ranking margin and $\zeta\geq0$ weights the negative
term. CompatNet sums this loss over all energies; $\zeta=0$ recovers ordinary
maximum-likelihood training. We use $(\chi,\zeta)=(0.5,0.2)$ for local
causal relations and $(1,1)$ for alarm-context energies.

\section{Mode-Aware Counterfactual Trajectory Optimization}
\label{app:repair}

This section defines the local proposal losses and alarm-context gradients
omitted from the main paper.

\subsection{Local Multi-Proposal Learning}
\label{app:repairnet-learning}

Using the proposal logits $\boldsymbol\eta_i^{(\ell)}$ from Eq.~(6) of the main
paper, the three losses are
\begin{equation}
\begin{aligned}
\mathcal L_r(\widetilde{\mathbf x}^{(1)},\mathbf x)
&=\frac1T\sum_{t=0}^{T-1}
\operatorname{CE}_i(\boldsymbol\eta_{i,t}^{(1)},y_{i,t}),\\
w^{(\ell)}
&=
\frac{\exp[-J(\widetilde{\mathbf x}^{(\ell)})/\kappa]}
{\sum_{\ell'=1}^{L}\exp[-J(\widetilde{\mathbf x}^{(\ell')})/\kappa]},\\
\mathcal L_e
&=\sum_{\ell=1}^{L}w^{(\ell)}
J(\widetilde{\mathbf x}^{(\ell)}),\\
\mathcal L_d
&=-\frac{2}{L(L-1)T}
\sum_{\ell<\ell'}\sum_{t=0}^{T-1}
\left|\widetilde x_{i,t}^{(\ell)}
-\widetilde x_{i,t}^{(\ell')}\right|.
\end{aligned}
\label{eq:supp-proposal-risk}
\end{equation}
Here $\operatorname{CE}_i$ is Bernoulli cross-entropy for a binary $X_i$ and
categorical cross-entropy on the native or binned support otherwise.
$\kappa>0$ is the Gibbs temperature. During normal-window training, $J$
includes all $A_j\in\mathcal A$, as specified after Eq.~(5) of the main paper.
$\mathcal L_e$ is a Gibbs-weighted risk rather than a log-sum-exp soft minimum,
and $\mathcal L_d$ is the negative mean pairwise distance between decoded
proposal trajectories. We use $\kappa=0.05$, $\alpha=1$, and $\beta=0.01$ in
Eq.~(7) of the main paper.

\subsection{Objective-Directed Refinement}
\label{app:refinement}

Local proposals are composed in graph order, so proposed parents condition
their children. Refinement then applies Eq.~(8) of the main paper only to logits associated with
$\mathcal M(R,\mathbf m)$; hard forward trajectories remain in
$\mathcal F(R,\mathbf m)$.

\noindent\textbf{Straight-through decoding.} For a binary scalar logit $\xi$ or
a nonbinary logit vector $\boldsymbol\xi$, define
\begin{equation}
\begin{aligned}
\boldsymbol\pi_{\tau_s}(\boldsymbol\xi)
&=
\begin{cases}
\bigl(1-\sigma(\xi/\tau_s),\,\sigma(\xi/\tau_s)\bigr)^\top,
& X_i\text{ binary},\\
\operatorname{softmax}(\boldsymbol\xi/\tau_s),
& X_i\text{ nonbinary},
\end{cases},\\
\mathbf y^{\mathrm{hard}}(\boldsymbol\xi)
&=
\begin{cases}
\operatorname{onehot}\!\left(\mathbf 1[\xi\geq0]\right),
& X_i\text{ binary},\\
\operatorname{onehot}(\arg\max_c\xi_c),
& X_i\text{ nonbinary},
\end{cases},\\
\widehat{\mathbf y}
&=\operatorname{sg}\!\left(
\mathbf y^{\mathrm{hard}}(\boldsymbol\xi)
-\boldsymbol\pi_{\tau_s}(\boldsymbol\xi)\right)
+\boldsymbol\pi_{\tau_s}(\boldsymbol\xi).
\end{aligned}
\label{eq:supp-ste}
\end{equation}
Here $\operatorname{sg}$ is stop-gradient. Thus
$\widehat{\mathbf y}=\mathbf y^{\mathrm{hard}}(\boldsymbol\xi)$ in the forward
pass, while its gradient is that of $\boldsymbol\pi_{\tau_s}$. The hard state is
decoded to its native value or fixed continuous-bin representative, yielding the trajectory
$\widetilde{\mathbf x}_s^{(\ell)}$ in Eq.~(8) of the main paper. Applied coordinate-wise,
Eq.~\eqref{eq:supp-ste} is $\operatorname{ST}_{\tau_s}$. The 50-step AdamW refinement uses
learning rate $0.2$, with $\tau_s$ linearly annealed from $1$ to $0.25$.

\noindent\textbf{Local-relation gradients.} Let
$\widehat{\mathbf y}_{i,t}$ be the STE state vector. The differentiable
embedding and target loss are
\begin{equation}
\begin{aligned}
\mathsf E(\widehat{\mathbf y}_{i,t})
&=\sum_c\widehat y_{i,t,c}\mathsf E(c),\\
\bar e_{i,t}
&=-\sum_c\widehat y_{i,t,c}
\log p_\theta(c\mid\mathbf h_{i,t}^C).
\end{aligned}
\label{eq:supp-local-relation-gradient}
\end{equation}
Here $\mathsf E$ is the applicable typed embedding in
Eq.~\eqref{eq:supp-rollout-state} or~\eqref{eq:supp-compat-routing}. The first
expression handles a mutable input, and the second a mutable child target.
Their hard forward values recover the likelihoods in
Section~\ref{app:local-likelihoods}, while Eq.~\eqref{eq:supp-ste} supplies the
gradients. Substituting these quantities into
Eqs.~\eqref{eq:supp-rollout-energy} and~\eqref{eq:supp-window-energy} yields the
surrogate local energy $\bar e_i^C$.

\noindent\textbf{Alarm-context gradients.} When parents of an active alarm are
refined, let $\widehat{\mathbf y}_{k,t}$ be the STE state of
$X_k\in\operatorname{pa}(A_j)$ (with proposal and refinement-step indices
suppressed). For joint configuration
$\mathbf c=(c_k)_{X_k\in\operatorname{pa}(A_j)}$, its STE mass is
\begin{equation}
\nu_{j,t}(\mathbf c)=
\prod_{X_k\in\operatorname{pa}(A_j)}
\widehat y_{k,t,c_k}.
\label{eq:supp-joint-relaxation}
\end{equation}
Writing $\boldsymbol\nu_{j,t}=(\nu_{j,t}(\mathbf c))_{\mathbf c}$, the
alarm-context surrogate is
\begin{equation}
\bar e_j^A=-\frac1T\sum_{t=0}^{T-1}\sum_{\mathbf c}
\nu_{j,t}(\mathbf c)
\log q_{\psi}(\mathbf c\mid\boldsymbol\nu_{j,0:t-1}).
\label{eq:supp-relaxed-alarm-energy}
\end{equation}
This is the multilinear STE extension of the same learned categorical
$q_\psi$ as in Eq.~(4) of the main paper. The forward pass uses hard one-hot
parent states and exactly recovers $e_j^A$; relaxed masses are used only to
differentiate with respect to parent logits.

\noindent\textbf{Refinement update.} The surrogate energies in
Eqs.~\eqref{eq:supp-local-relation-gradient}
and~\eqref{eq:supp-relaxed-alarm-energy} form the update below. Reinstating
the suppressed indices, let \(\bar e_{i,s}^{C,(\ell)}\) and
\(\bar e_{j,s}^{A,(\ell)}\) denote these energies for proposal \(\ell\) at
refinement step \(s\):
\begin{equation}
\begin{aligned}
\bar J_s^{(\ell)}
&=\sum_{A_j\in\mathcal I\cap\mathcal A^+}
[\bar e_{j,s}^{A,(\ell)}-\tau_j^A]_+\\
&\quad+\lambda\sum_{X_i\in\mathcal I\cap\mathcal X}
[\bar e_{i,s}^{C,(\ell)}-\tau_i^C]_+,\\
\boldsymbol\xi_{s+1}^{(\ell)}
&=\boldsymbol\xi_s^{(\ell)}
-\eta_s\nabla_{\boldsymbol\xi_s^{(\ell)}}\bar J_s^{(\ell)}.
\end{aligned}
\label{eq:supp-refinement-update}
\end{equation}
Hard forward decoding gives
$\bar J_s^{(\ell)}=J(\widetilde{\mathbf x}_s^{(\ell)})$ numerically; the bar
only marks the soft backward derivative. AdamW uses this gradient, whereas
proposal selection and final scoring use the original hard objective $J$.

\section{Neural Network Architectures}
\label{app:architectures}

This section records the parameter-sharing and routing choices not implied by
standard GRU and dilated residual TCN blocks
\citep{cho2014learning,bai2018empirical}. Full-window and RepairNet TCNs use
symmetric padding; the alarm-context TCN uses left padding for causal scoring.
Figure~\ref{fig:supp-neural-architectures} summarizes both architectures
and their routing boundaries.

\begin{figure*}[!t]
  \centering
  \includegraphics[width=0.99\textwidth]{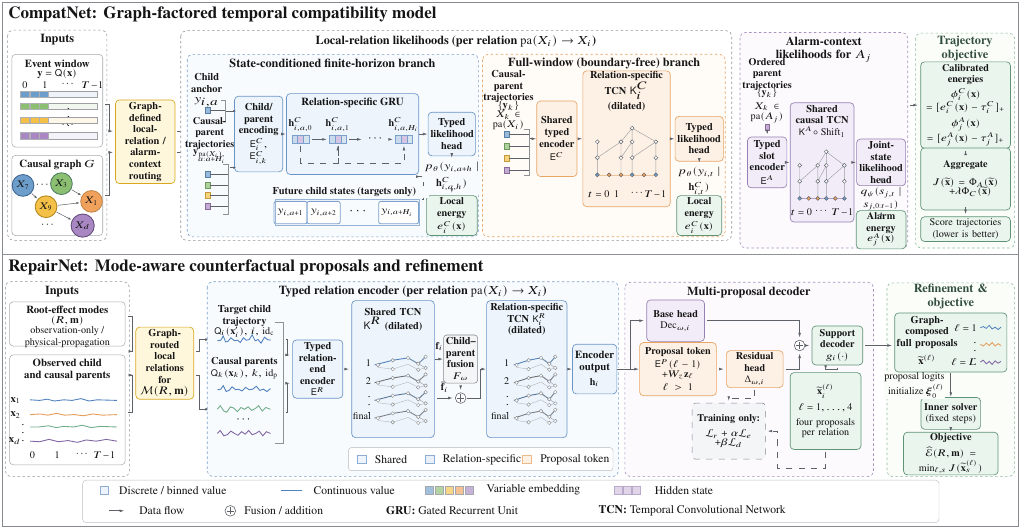}
  \caption{CompatNet and RepairNet architectures. Top: graph-routed
  local-relation banks and the alarm-context branch produce calibrated energies
  that form $J(\widetilde{\mathbf x})$. Bottom: $(R,\mathbf m)$ determines
  $\mathcal M(R,\mathbf m)$, and four full proposals initialize refinement,
  yielding the fixed-inner-solver approximate value
  $\widehat{\mathcal E}(R,\mathbf m)
  =\min_{\ell,s}J(\widetilde{\mathbf x}_s^{(\ell)})$.
  Dashed paths denote training-only supervision.}
  \label{fig:supp-neural-architectures}
  \label{fig:supp-compatnet}
  \label{fig:supp-repairnet}
\end{figure*}

\subsection{CompatNet}
\label{app:compatnet-architecture}

\paragraph{State-conditioned local-relation bank.}
Each local causal relation routed to this bank has an independently
parameterized GRU and typed likelihood head implementing
Eq.~\eqref{eq:supp-rollout-state}. Child and parent state
embeddings initialize and drive the recurrence; future child states remain
likelihood targets and are never fed back.

\paragraph{Full-window local-relation bank.}
This branch implements Eq.~\eqref{eq:supp-window-energy} with its own shared
typed encoder $\mathsf E^C$ and a relation-specific TCN. For parent $X_k$,
$\mathsf E^C(\mathbf y_k,k)$ encodes its state sequence, identity, and support
type; the resulting context is
\begin{equation}
\mathbf h_i^C=
\mathsf K_i^C\!\left(
\frac{\sum_{X_k\in\operatorname{pa}(X_i)}\mathsf E^C(\mathbf y_k,k)}
{|\operatorname{pa}(X_i)|}\right).
\label{eq:supp-compat-routing}
\end{equation}
Thus $i$ selects an independently parameterized TCN rather than acting as a
learned embedding. At time $t$, the relation-specific head applies
Eq.~\eqref{eq:supp-typed-heads} to $\mathbf h_{i,t}^C$, yielding
$\pi_{\theta,i,t}$ or $\boldsymbol\pi_{\theta,i,t}$.

\paragraph{Alarm-context branch.}
For $X_k\in\operatorname{pa}(A_j)$, let $\mathsf E^A(\mathbf y_k,k)$ encode
its state sequence, identity, support type, and fixed slot. After aggregation,
$\operatorname{Shift}_1$ right-shifts the sequence by one step and prepends a
learned initial vector before causal encoding:
\begin{equation}
\mathbf h_j^A=
\mathsf K^A\!\left(
\operatorname{Shift}_1\!\left[
\frac{1}{|\operatorname{pa}(A_j)|}
\sum_{X_k\in\operatorname{pa}(A_j)}
\mathsf E^A(\mathbf y_k,k)\right]\right).
\label{eq:supp-alarm-routing}
\end{equation}
The causal TCN $\mathsf K^A$ is shared, while the alarm-specific categorical
head supplies $q_\psi(s_{j,t}\mid s_{j,0:t-1})$ in
Eq.~(4) of the main paper over
$\prod_{X_k\in\operatorname{pa}(A_j)}C_k$ joint states. The shift keeps
the current target out of the input; the ordered parents and head identify the
alarm without an alarm-ID embedding.

\subsection{RepairNet}
\label{app:repairnet-architecture}

\paragraph{Typed relation encoder.}
The encoder $\mathsf E^R$ combines a trajectory with its variable identity,
support type, and a learned relation-end identity. Let
$\mathrm{id}_{\mathrm c}$ and $\mathrm{id}_{\mathrm p}$ denote the child- and
parent-end identities, distinct from variable identities $i$ and $k$. A shared
TCN $\mathsf K^R$ processes each token independently. For local causal relation
$(X_i,\operatorname{pa}(X_i))$, define
\begin{equation}
\begin{aligned}
\mathbf f_i
&=\mathsf K^R\!\left(
\mathsf E^R(\mathsf Q_i(\mathbf x_i'),i,\mathrm{id}_{\mathrm c})\right),\\
\bar{\mathbf f}_i
&=\frac{1}{|\operatorname{pa}(X_i)|}
\sum_{X_k\in\operatorname{pa}(X_i)}
\mathsf K^R\!\left(
\mathsf E^R(\mathsf Q_k(\mathbf x_k),k,\mathrm{id}_{\mathrm p})\right),\\
\mathbf h_i
&:=\mathsf K_i^R\!\left(
\operatorname{LN}\!\left[
\operatorname{LN}(\mathbf f_i)+\bar{\mathbf f}_i+
F_\omega[\mathbf f_i;\bar{\mathbf f}_i]\right]\right).
\end{aligned}
\label{eq:supp-repair-encoder}
\end{equation}
Here $F_\omega$ is a two-layer MLP. This construction instantiates
$\operatorname{Enc}_\omega$ in Eq.~(6) of the main paper, and child identity
$i$ selects $\mathsf K_i^R$ for relation
$\operatorname{pa}(X_i)\to X_i$.

\paragraph{Multi-proposal heads.}
The deterministic $\operatorname{Dec}_{\omega,i}$ in Eq.~(6) of the main paper
is a relation-specific typed linear projection with the output dimensions in
Eq.~\eqref{eq:supp-typed-heads}. For $\ell>1$,
$\Delta_{\omega,i}$ receives
$\operatorname{LN}[\mathbf h_i+\mathsf E^P(\ell-1)+W_z\mathbf z_\ell]$, applies two
relation-specific kernel-$5$ TCN blocks with dilations $(1,2)$, and projects to
the same typed output.

$\mathsf E^P$ is a learned proposal-index embedding and $W_z$ projects
proposal noise. Both are shared, whereas
$\operatorname{Dec}_{\omega,i}$ and $\Delta_{\omega,i}$ are relation-specific.
Each $\mathbf z_\ell\sim\mathcal N(\mathbf0,I_8)$ is sampled once per window
and broadcast over time, where $I_8$ is the $8\times8$ identity. We use $L=4$
and decoder $g_i$ from Eq.~(6) of the main paper.

\subsection{Dimensions and Training}
\label{app:network-training}

\noindent\textbf{Architecture and routing.}
Typed heads retain the support-dependent output sizes in
Eq.~\eqref{eq:supp-typed-heads}; the remaining settings are:
\begin{center}
\footnotesize
\setlength{\tabcolsep}{2pt}
\renewcommand{\arraystretch}{1.04}
\begin{tabular}{@{}>{\raggedright\arraybackslash}p{0.23\columnwidth}
                    >{\raggedright\arraybackslash}p{0.26\columnwidth}
                    >{\raggedright\arraybackslash}p{0.43\columnwidth}@{}}
\toprule
Path & Assigned relations & Width and temporal block \\
\midrule
State-conditioned bank
& Configured state-conditioned relations; precedence on overlap
& 8-d embeddings; 32-d relation-specific GRU \\
Full-window bank
& Remaining non-source relations
& 128-d encoder/state; kernel-5 TCN with $(1,2,4)$ \\
Alarm-context branch
& Joint parent context of each alarm
& 64-d encoder/state; causal kernel-5 TCN with
  $(1,2,4,8,16,32)$ \\
RepairNet
& One relation-specific decoder per target child
& 64-d encoder; 8/64-d proposal conditioning; pre/post-fusion
  dilations $(1,2,4)$ and $(1,2,4,8,16,32)$ \\
\bottomrule
\end{tabular}
\end{center}

\noindent\begin{minipage}{\columnwidth}
\noindent\textbf{Training protocol.}
Normal sequences are split before extracting overlapping windows or rollout
segments. Rows are executed top to bottom; refinement is inference-only.
\begin{center}
\footnotesize
\setlength{\tabcolsep}{2pt}
\renewcommand{\arraystretch}{1.02}
\begin{tabular}{@{}>{\raggedright\arraybackslash}p{0.18\columnwidth}
                    >{\raggedright\arraybackslash}p{0.34\columnwidth}
                    >{\centering\arraybackslash}p{0.19\columnwidth}
                    >{\centering\arraybackslash}p{0.18\columnwidth}@{}}
\toprule
Stage & Training unit; validation criterion & Batch / epochs
& AdamW lr/wd \\
\midrule
Full-window bank
& Relation--window; validation NLL; relation-balanced
& 6144/80 & $10^{-3}/10^{-4}$ \\
State-conditioned bank
& Relation--segment; macro validation NLL; regime-balanced
& 512/100 & $10^{-3}/10^{-3}$ \\
Alarm-context branch
& Window; validation NLL
& 256/60 & $10^{-3}/10^{-4}$ \\
Calibrate and freeze
& Held-out energies; empirical relation- and alarm-specific quantiles
& -- & -- \\
RepairNet
& Pseudo-corrupted target--relation window; validation proposal criterion
& 1024/50 & $10^{-3}/10^{-4}$ \\
Refinement
& Event-time optimization of Eq.~(1) of the main paper
& -- & Inference only \\
\bottomrule
\end{tabular}
\end{center}
Fault labels, root annotations, and test data or metrics are excluded from
training, calibration, and checkpoint selection.
\end{minipage}

\noindent\begin{minipage}{\columnwidth}
\noindent\textbf{Dataset-specific settings.}
Entries list only overrides; pairs denote batch size/max.\ epochs.
\begin{center}
\footnotesize
\setlength{\tabcolsep}{2pt}
\renewcommand{\arraystretch}{1.04}
\begin{tabular}{@{}>{\raggedright\arraybackslash}p{0.20\columnwidth}
                    >{\raggedright\arraybackslash}p{0.72\columnwidth}@{}}
\toprule
Dataset & Overrides \\
\midrule
LM
& Full-window $2752/100$; alarm-context branch $256/200$;
  RepairNet $768/30$ \\
TS
& RepairNet $1024/100$ \\
QT
& RepairNet $256/100$; four tank levels use the state-conditioned bank
  ($30$-s horizon, $5$-s stride); post-fusion dilations $(1,2,4)$ \\
causRCA Full
& RepairNet $256/100$ with lr \(3{\times}10^{-4}\) \\
\bottomrule
\end{tabular}
\end{center}
\end{minipage}

\section{Certified Residual-Cover Best-Bound Search}
\label{app:certified-search}

This section completes the theoretical development of the search in the main
paper. It fixes the event-specific mode space and inner solver, proves
Proposition~1, gives the exact binary encoding required by Proposition~2, and
then proves Theorem~1 and Corollary~1. The final subsection records the
conditions under which the implementation reports a certificate.

\subsection{Fixed Mode Space, Inner Solver, and Lower Bound}
\label{app:solver-definitions}

We retain $\mathcal R_0$, $\mathcal M_{r,m}$, $S_v$, $c_v$, $L(R,\mathbf m)$,
$\mathcal E$, $\widehat{\mathcal E}$, $U$, and $\widehat U$ from the main
paper. Here \(\mathcal E\) is the ideal constrained inner value and \(U\) is
the ideal outer root--mode objective, whereas
$\widehat{\mathcal E}(R,\mathbf m)
=J(\widehat{\mathbf x}_{R,\mathbf m})$ and
$\widehat U(R,\mathbf m)=\widehat{\mathcal E}(R,\mathbf m)+\gamma|R|$ are their fixed-inner-solver
counterparts. Here $\varepsilon$ is the propagation
threshold in $\mathcal M_{r,m}$, whereas $\epsilon$ is the outer-search
optimality tolerance. For each event, the scopes
$\mathcal M_{r,m}$ are computed from the observed trajectory and then frozen
throughout search. The admissible root--mode actions are exactly those enabled by
the inference configuration, each with its frozen scope $\mathcal M_{r,m}$.
The MILP and inner solver use the same actions; both modes are enabled in the
reported runs.

Fixing the checkpoints, inference settings, realized proposal-noise draws, and
optimizer schedule makes the inner solver deterministic. It returns
$\widehat{\mathbf x}_{R,\mathbf m}\in\mathcal F(R,\mathbf m)$, and a root set is
excluded only after the solver evaluates all its admissible mode assignments.
Indexed proposal-noise draws are fixed before search and reused whenever an
assignment is evaluated.
One \emph{root--mode evaluation} is one inner-solver call for a single
$(R,\mathbf m)$; a root-set evaluation is complete only after all admissible
$\mathbf m$ have been evaluated.
For $v\in\mathcal I$, the corresponding energy depends only on $S_v$. If
$S_v\cap\mathcal M(R,\mathbf m)=\varnothing$, feasibility fixes every argument
of that energy, so its observed contribution $c_v$ is retained. Otherwise, the
residual-cover bound replaces the potentially changeable, nonnegative
contribution by zero.

\begin{proof}[Proof of Proposition 1]
Fix $(R,\mathbf m)$ and any
$\widetilde{\mathbf x}\in\mathcal F(R,\mathbf m)$. If
$S_v\cap\mathcal M(R,\mathbf m)=\varnothing$, then
$\widetilde{\mathbf x}_{S_v}=\mathbf x_{S_v}$, so the corresponding energy term
contributes $c_v$. Every remaining calibrated energy term is nonnegative.
Summing the retained terms and the exact cardinality term gives
$J(\widetilde{\mathbf x})+\gamma|R|\ge L(R,\mathbf m)$. Minimizing over
$\mathcal F(R,\mathbf m)$ proves $L(R,\mathbf m)\le U(R,\mathbf m)$. The fixed
inner-solver output is feasible, so
$U(R,\mathbf m)\le\widehat U(R,\mathbf m)$.
\end{proof}

\subsection{Exact MILP Encoding}
\label{app:residual-cover}

For $r\in\mathcal C$ and $m\in\{\mathrm{o},\mathrm{p}\}$, binary
$\upsilon_{r,m}$ selects a root--mode action,
$\upsilon_r=\upsilon_{r,\mathrm{o}}+\upsilon_{r,\mathrm{p}}$ indicates root
membership, and $\iota_v$ retains $c_v$. Inadmissible root--mode actions are fixed
to zero. At iteration $t$, the exact binary encoding of Eq.~(10) of the main paper is
\begingroup
\small
\begin{equation}
\begin{aligned}
\min_{\upsilon,\iota}\quad&
\gamma\!\sum_{r\in\mathcal C}\upsilon_r+
\sum_{v\in\mathcal I}c_v\iota_v\\
\mathrm{s.t.}\quad&
\iota_v+\!\!\sum_{\substack{r\in\mathcal C,\,m\in\{\mathrm{o},\mathrm{p}\}:\\
S_v\cap\mathcal M_{r,m}\ne\varnothing}}
\upsilon_{r,m}\geq1,\quad\forall v\in\mathcal I,\\
&\upsilon_r=\upsilon_{r,\mathrm{o}}+\upsilon_{r,\mathrm{p}},
\quad\forall r\in\mathcal C,\\
&1\leq\sum_{r\in\mathcal C}\upsilon_r\leq K,\\
&\sum_{r\in\operatorname{an}(A_j)}\upsilon_r\geq1,\quad
\forall A_j\in\mathcal A^\top,\\
&\sum_{r\in R_s}(1-\upsilon_r)+
\sum_{r\in\mathcal C\setminus R_s}\upsilon_r\geq1,
\quad\forall s<t,\\
&\upsilon_{r,m}=0,\quad
\forall (r,m)\text{ inadmissible},\\
&\upsilon_{r,m},\upsilon_r,\iota_v\in\{0,1\}.
\end{aligned}
\label{eq:supp-residual-cover-milp}
\end{equation}
\endgroup

\begin{lemma}[Exact retained-energy objective]
Fix a feasible $(R,\mathbf m)$ and its corresponding $\upsilon$ variables.
Minimizing the calibrated-energy part of
Eq.~\eqref{eq:supp-residual-cover-milp} over $\iota$ yields the retained-energy
sum in $L(R,\mathbf m)$.
\end{lemma}

\begin{proof}
No selected action intersects $S_v$ exactly when
$S_v\cap\mathcal M(R,\mathbf m)=\varnothing$. The first constraint then forces
$\iota_v=1$. If an action does intersect $S_v$, setting $\iota_v=0$ is feasible
and minimizes its nonnegative coefficient $c_v$. Thus the optimal energy
contribution is
$c_v\mathbf 1[S_v\cap\mathcal M(R,\mathbf m)=\varnothing]$; summing over
$v\in\mathcal I$ proves the claim.
\end{proof}

\begin{proof}[Proof of Proposition 2]
Because all variables are binary,
$\upsilon_r=\upsilon_{r,\mathrm{o}}+\upsilon_{r,\mathrm{p}}$ selects at most one
admissible mode per root. The cardinality and alarm-coverage rows therefore map
every feasible $\upsilon$ to an assignment with $R\in\mathcal R_0$; conversely,
every admissible assignment in $\mathcal R_0$ defines such a vector.

For an evaluated root set $R_s$, its root-set exclusion constraint equals
$|R_s\setminus R|+|R\setminus R_s|=|R\triangle R_s|$. It is zero only for
$R=R_s$, so these constraints remove exactly
$R_0,\ldots,R_{t-1}$, independently of mode, and leave $\mathcal R_t$.
The lemma makes the objective for each remaining assignment exactly
$L(R,\mathbf m)$. Hence the MILP optimum is
$\min_{R\in\mathcal R_t}\min_{\mathbf m}L(R,\mathbf m)=B_t$, and its selected
root set is a valid $R_t$ in Eq.~(10) of the main paper.
\end{proof}

\subsection{Best-Bound Certificates}
\label{app:best-bound-proof}

\begin{proof}[Proof of Theorem 1]
For nonempty $\mathcal R_t$, Propositions~1--2 give
\[
\begin{aligned}
B_t
&=\min_{R\in\mathcal R_t}\min_{\mathbf m}L(R,\mathbf m)\\
&\leq\min_{R\in\mathcal R_t}\min_{\mathbf m}
\widehat U(R,\mathbf m)
=\min_{R\in\mathcal R_t}\widehat U_R.
\end{aligned}
\]
For $t\geq1$, after at least one completed root-set evaluation,
$\widehat U_t$ is finite and
\[
\begin{aligned}
\widehat U_t-\widehat U^\star
&=\left[\widehat U_t-
\min_{R\in\mathcal R_t}\widehat U_R\right]_+\\
&\leq[\widehat U_t-B_t]_+,
\end{aligned}
\]
which proves the anytime gap. At $t=0$, only
$B_0\leq\min_{R\in\mathcal R_0}\widehat U_R$ is asserted; the finite anytime
gap starts at $t=1$.
If $\mathcal R_t=\varnothing$, all root sets have been evaluated and optimality
is immediate under the convention $B_t=+\infty$.

If $B_t\geq\widehat U_t$, the gap is zero and
$\widehat U_t=\widehat U^\star$. Under the stopping rule
$B_t>\widehat U_t+\epsilon$, every unseen set satisfies
$\widehat U_R\geq B_t>\widehat U^\star+\epsilon$, so every
$\epsilon$-optimal root set has already been evaluated. Each completed
iteration removes one previously unseen root set only after all its admissible
modes have been evaluated. Therefore at most $|\mathcal R_0|$ completed
root-set evaluations are required.
\end{proof}

\begin{proof}[Proof of Corollary 1]
Certification gives $\widehat U_t=\widehat U^\star$. Every different evaluated
root set has objective at least $\widehat U_t^{\mathrm{alt}}$, while every
unseen root set has objective at least $B_t$. Hence
$\widehat U_R\geq\min\{\widehat U_t^{\mathrm{alt}},B_t\}$ for all
$R\ne\widehat R^\star$. Moreover,
$\widehat U_t^{\mathrm{alt}}\geq\widehat U_t$ and
$B_t\geq\widehat U_t$ at certification, so $\underline\Delta_t\geq0$.
Subtracting $\widehat U^\star=\widehat U_t$ gives the stated lower bound on
every alternative root set. If $\underline\Delta_t>\epsilon$, all alternatives
lie outside the $\epsilon$-optimal set, proving uniqueness.
\end{proof}

\subsection{Implementation and Certificate Scope}
\label{app:numerical-certificate}

The base MILP can be strengthened when calibrated energy terms, mutable scopes,
repair dependencies, and fixed proposal-noise draws share a strict component
partition $\mathcal P$. Let $\mathbf s_P(R,\mathbf m)$ record every candidate
root in $P$ as unselected, $\mathrm{o}$, or $\mathrm{p}$. The verified
decomposition has the form
\[
\begin{aligned}
\widehat U(R,\mathbf m)
&=c_0+\sum_{P\in\mathcal P}
\widehat U_P(\mathbf s_P(R,\mathbf m)),\\
L(R,\mathbf m)
&=c_0+\sum_{P\in\mathcal P}
L_P(\mathbf s_P(R,\mathbf m)),
\end{aligned}
\]
where $c_0$ collects immutable terms, with each root penalty assigned to one
component. Before iteration $t$, let $\mathcal H_{P,t}$ contain states
evaluated by the fixed inner solver restricted to $P$. Provided
$L_P(\mathbf s)\leq\widehat U_P(\mathbf s)$ for every cached state, define
\[
\begin{aligned}
L_t^+(R,\mathbf m)
=L(R,\mathbf m)
&+\sum_{P\in\mathcal P}\sum_{\mathbf s\in\mathcal H_{P,t}}
\mathbf 1[\mathbf s_P(R,\mathbf m)=\mathbf s]\\
&\qquad\cdot
\bigl(\widehat U_P(\mathbf s)-L_P(\mathbf s)\bigr).
\end{aligned}
\]
At most one cached state matches per component; a binary conjunction activates
that match, so the augmented MILP minimizes $L_t^+$ exactly.
The decomposition gives $L_t^+(R,\mathbf m)\leq\widehat U(R,\mathbf m)$, although
$L_t^+$ need not lower-bound the ideal $U$. Proposition~1 therefore concerns
the base $L$; Theorem~1 and Corollary~1 remain valid with
$B_t^+=\min_{R\in\mathcal R_t}\min_{\mathbf m}L_t^+(R,\mathbf m)$.

We solve the base or strengthened binary MILP with SciPy/HiGHS
\citep{huangfu2018parallelizing}.

All certificates concern the finite admissible root--mode space for the
provided graph, frozen checkpoints, and fixed inner-solver settings and
proposal-noise draws. They do not certify the global optimum of the continuous inner problem,
the correctness of the learned compatibility model, or causal completeness
outside $\mathcal X$.

\section{Datasets, Graphs, and Event Representation}
\label{app:dataset-details}

All five groups use directed causal graphs and time-aligned event trajectories;
Figure~\ref{fig:supp-dataset-form} shows three representative events per group.
Figures~\ref{fig:supp-causrca-graphs} and~\ref{fig:supp-synthetic-graphs}
expose the graph structure, while
Table~\ref{tab:supp-dataset-statistics} separates graph scale, data types, and
event counts. The causRCA graph is the provided industrial expert graph:
\emph{Probe}, \emph{Coolant}, and \emph{Hydraulics} are induced views with
11/15/17 nodes and 34/25/41 events, and \emph{Full} contains all 92 nodes over
the same 100 physical events. Thus the four causRCA settings are reporting
views, not four independent event collections. Figure~
\ref{fig:supp-causrca-graphs} separates the three subsystem views; the
redundant \emph{Full} rendering is omitted.

The four synthetic groups stress different aspects of RCA. TA tests delays in
binary command--response chains; TS tests mixed-type signals by combining
continuous measurements with categorical operating states; LM tests scalability
by repeating a binary module 50 times (206 nodes); and QT tests coupled nonlinear
dynamics using the quadruple-tank process \citep{johansson2000quadruple}. The
next subsections specify how clean trajectories are simulated, how
observation-only and physical-propagation faults are inserted, and when
downstream variables and alarms are recomputed.

\begin{table*}[t]
\centering
\small
\setlength{\tabcolsep}{2.1mm}
\begin{tabular}{@{}llrrlrr@{}}
\toprule
Group & Graph source & Nodes & Edges & Signal type mix & Normal runs & Test events \\
\midrule
causRCA & provided industrial expert graph & 92 & 104 & 78 Bool., 3 Int., 8 Cont., 3 Cat. & 170 & 100 (four views) \\
TA & generator graph & 10 & 21 & 10 Bool. & 48 & 30 \\
TS & generator graph & 10 & 12 & 5 Bool., 2 Cont., 3 Cat. & 48 & 30 \\
LM & generator graph & 206 & 400 & 206 Bool. & 32 & 50 \\
QT & physics-derived graph & 24 & 42 & 7 Bool., 8 Cont., 9 Cat. & 112 & 35 \\
\bottomrule
\end{tabular}
\caption{Dataset and graph statistics. Alarm and trusted-context nodes are
included in the node and signal-type counts.}
\label{tab:supp-dataset-statistics}
\end{table*}

\begin{figure*}[!t]
  \centering
  \includegraphics[width=\textwidth]{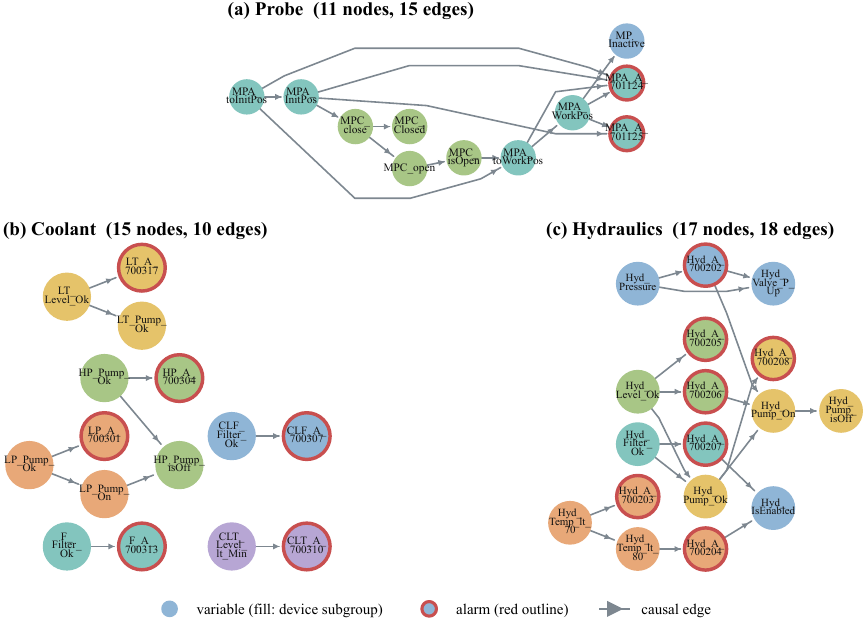}
  \caption{Provided causRCA subsystem graphs in a compact left-to-right
  directed layout. Fill colors visually group device prefixes; red outlines
  mark alarms.}
  \label{fig:supp-causrca-graphs}
\end{figure*}

\subsection{Fault Injection into Clean Trajectories}
\label{app:synthetic-generation}

Each fault event is paired with a clean trajectory generated under the same
commands, contexts, initial conditions, and random draws. The faulty copy
therefore differs only by the injected fault. An observation-only fault changes
only the recorded trajectory. A physical-propagation fault changes the
underlying state, after which descendants and alarms are recomputed in causal
order.

TA, TS, and LM use a 0.25\,s grid and 240\,s normal runs. Each event has a
55\,s simulation horizon; its fault starts on-grid at \(t_c\in[25,30]\) seconds
and lasts 8--10\,s. The saved record spans \(t=0\) through
\(t_a+10\,\mathrm{s}\), where \(t_a\) is the latest target-alarm onset. It
therefore retains 25--30\,s of clean pre-fault behavior and 10\,s after all
target alarms have appeared. QT uses the same output grid, 120\,s normal
episodes, and 180\,s event simulations with \(t_c=30\,\mathrm{s}\). Its saved
window is
\([t_c-15\,\mathrm{s},
\max\{t_a+10\,\mathrm{s},t_c+15\,\mathrm{s}\}]\), which contains at least
30\,s.

Algorithm~\ref{alg:supp-data-generation} summarizes the shared procedure. The
paired clean files and simulator truth are used only to audit generation and
analyze failures in Section~\ref{app:failure-cases}; inference receives only
the faulty event record. Generator-side injection types distinguish
observation-only from physical-propagation generation and are not
ground-truth labels for the auxiliary root-effect modes inferred by the method.

\begin{algorithm}[!t]
\small
\caption{Fault injection into paired clean trajectories}
\label{alg:supp-data-generation}
\begin{algorithmic}[1]
\Require graph \(G\), causal mechanisms, alarm rules, fault catalog,
grid \(\Delta t\)
\State Generate normal runs from sampled contexts and commands
\For{each fault event}
  \State sample clean \(\mathbf x^0\), onset \(t_c\), root set \(R\), and injection types
  \State \(\mathbf x\gets\mathbf x^0\)
  \For{\(r\in R\)}
    \If{\(r\)'s injection type is observation-only}
      \State modify only the recorded trajectory \(\mathbf x_r\)
    \Else
      \State modify \(r\)'s state; recompute its descendants
    \EndIf
  \EndFor
  \State recompute alarms; let \(t_a\) be the latest target-alarm onset
  \State crop the group-specific window
  \State verify \(\mathbf x^0\) is alarm-free and \(R\) reaches each target alarm
  \State verify those alarms disappear when every
  \(\mathbf x_r,\ r\in R\), is restored
  \State write \(\mathbf x\) in long form
\EndFor
\end{algorithmic}
\end{algorithm}

\subsection{Generation Equations by Synthetic Group}

\paragraph{TA: binary command--response delays.}
Let \(D_{[p,q]}(z)\) apply a uniform \([p,q]\)-step delay to each binary
transition:
\begin{equation}
\begin{aligned}
\mu_i^{\mathrm{mon}}&=D_{[0,1]}(u_i),&
x_1&=D_{[1,4]}(u_1),\\
g&=D_{[1,3]}(u_1x_1),&
x_2&=D_{[1,4]}(u_2x_1g).
\end{aligned}
\label{eq:supp-ta-generator}
\end{equation}
Here \(u_i\), \(x_i\), and \(g\) are the TA command, response, and guard
trajectories, respectively, with time arguments suppressed, while
\(\mu_i^{\mathrm{mon}}\) is a command monitor. Commands cycle through
idle--clamp--transfer--hold. For
\(r\in\{\mathrm{motion},\mathrm{guard},\mathrm{sequence}\}\), let \(v_r(t)\)
be its binary violation trace. Define the sustained-alarm operator
\begin{equation}
\begin{aligned}
S_{w,\theta}[v](t_j)
&:=\mathbf 1\!\left[
\frac1W\sum_{\ell=0}^{W-1}v(t_{j-\ell})\geq\theta\right],
\quad W=\frac{w}{\Delta t},\\
a_r(t_j)&:=S_{W_r\Delta t,\theta_r}[v_r](t_j).
\end{aligned}
\label{eq:supp-alarm-generator}
\end{equation}
with \((W_r\Delta t,\theta_r)=(4,0.65),(4,0.60),(3,0.55)\) in the listed order,
and alarms persist for at least 6\,s. The 30 events cover stuck, slow,
short-pulse, and joint observation-only/physical-propagation faults.

\paragraph{TS: operating regimes and sensor responses.}
The process regimes are idle, fill, drain, recirculate, and pressurize; commands lag
0.25--0.75\,s. The endpoints are
\begin{equation}
\begin{aligned}
\mu_P(t)&=0.50+0.04\,\mathrm{PumpOn}(t)-0.04\,\mathrm{ValveOut}(t),\\
\mu_L(t)&=0.50+0.16\,\mathrm{ValveIn}(t)-0.16\,\mathrm{ValveOut}(t).
\end{aligned}
\label{eq:supp-ts-endpoints}
\end{equation}
Raw pressure and level track \(\mu_P,\mu_L\) through a one-step-delayed 2\,s
exponential response with 0.04 and 0.035 overshoot. Thresholds 0.42 and 0.58
define low/normal/high states. Equation~\ref{eq:supp-alarm-generator} raises
PressureAlarm outside normal and LevelAlarm after a 3\,s Mode--LevelState
mismatch. Events cover physical-propagation sources, observation-only faults,
and three boundary cases.

\paragraph{LM: repeated command--response modules.}
For module \(k\), \(u_k\), \(x_k\), and \(g_k\) denote its command, response,
and guard trajectories, with time arguments suppressed. For ten-module group
\(\mathcal G_j\),
\begin{equation}
\begin{aligned}
x_k&=D_{\delta(q_{\mathrm{LM}})}(u_k),&
g_k&=D_{[1,3]}(x_k),\\
a_k&=S_{3\,\mathrm{s},0.65}
\bigl[|u_k-x_k|\lor|x_k-g_k|\bigr],\\
a_j^{\mathrm{grp}}&=S_{1\,\mathrm{s},0.25}
\left[\max_{k\in\mathcal G_j}a_k\right].
\end{aligned}
\label{eq:supp-lm-generator}
\end{equation}
where \(q_{\mathrm{LM}}\) is the trusted shared utility context shown as node
\(q\) in Figure~\ref{fig:supp-synthetic-graphs}, and
\(\delta(q_{\mathrm{LM}})=[1,3]\) when \(q_{\mathrm{LM}}=1\) and \([3,6]\)
otherwise. The operator \(S_{w,\theta}\) is defined in
Eq.~\eqref{eq:supp-alarm-generator}. The 10/25/15
one-/two-/three-root events use distinct groups and activate both alarm levels.

\paragraph{QT: quadruple-tank dynamics and alarms.}
For tank levels \(h_i\), fixed-step RK4 integrates
\begin{equation}
\begin{aligned}
\dot h_1&=-\frac{a_1}{A_1}\sqrt{2gh_1}
 +\frac{a_3}{A_1}\sqrt{2gh_3}
 +\frac{\gamma_1k_1v_1}{A_1},\\
\dot h_2&=-\frac{a_2}{A_2}\sqrt{2gh_2}
 +\frac{a_4}{A_2}\sqrt{2gh_4}
 +\frac{\gamma_2k_2v_2}{A_2},\\
\dot h_3&=-\frac{a_3}{A_3}\sqrt{2gh_3}
 +\frac{(1-\gamma_2)k_2v_2}{A_3},\\
\dot h_4&=-\frac{a_4}{A_4}\sqrt{2gh_4}
 +\frac{(1-\gamma_1)k_1v_1}{A_4},
\end{aligned}
\label{eq:supp-qt-generator}
\end{equation}
\[
\begin{aligned}
A&=(28,32,28,32)\,\mathrm{cm}^2,\\
a&=(0.071,0.057,0.071,0.057)\,\mathrm{cm}^2,\\
k&=(3.33,3.35)\,\mathrm{cm}^3/(\mathrm{V\,s}),\\
g&=981\,\mathrm{cm}/\mathrm{s}^2.
\end{aligned}
\]
Here \(v_i,\gamma_i\) are pump and split states. Internal/output steps are
0.05/0.25\,s; pump/split time constants are 0.35/0.50\,s. Branch alarms require
a 2\,s mismatch above 0.05\,V or 0.03 split ratio; tank alarms require a 4\,s
deviation above 0.12\,cm outside transition grace. Thirty events follow trained
transition classes; five are out of support.

\begin{figure*}[!t]
  \centering
  \includegraphics[width=0.99\textwidth]{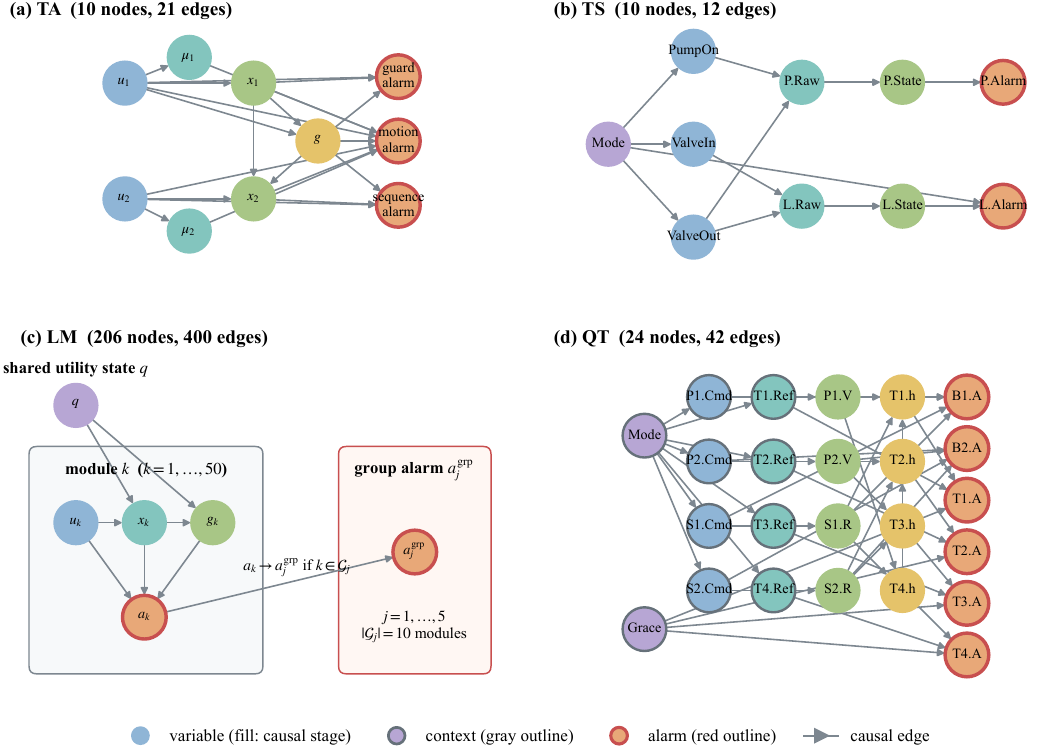}
  \caption{Synthetic causal structures. TA, TS, and QT show the full graphs;
  LM contracts its repeated modules, with counts referring to the full graph.
  In TA, \(\mu_i\) denotes a command monitor. In TS, P/L denote
  pressure/level, followed by Raw, State, or Alarm. In QT, P/S/T/B denote
  pump, split, tank, and branch; Cmd, Ref, V, R, \(h\), A, Mode, and Grace
  denote command, reference, voltage, ratio, level, alarm, operating mode, and
  transition grace.}
  \label{fig:supp-synthetic-graphs}
\end{figure*}

\section{Baseline Adapters and Evaluation Protocol}
\label{app:temporal-baselines}

\subsection{Method Adapters}

The dagger denotes an interface adaptation required by the event data; each
method's original inference rule is retained.
\begin{itemize}
\setlength{\itemsep}{2pt}
\item \textbf{IDI$^\dagger$.} IDI fits an SCM for a varying
service-level objective/key performance indicator (SLO/KPI) and scores
interventions by whether they restore that target~\citep{nagalapatti2025robust}.
Our alarm target is always 0 in normal runs, yielding a constant model that
makes nearly any intervention appear successful. IDI$^\dagger$ instead tests
whether the intervened alarm-parent tuple returns to normal support, while
retaining IDI's joint-intervention and Shapley attribution.

\item \textbf{AERCA$^\dagger$.} A fixed mask sets lagged coefficients outside
the provided causal graph to zero in AERCA's prediction, sparsity, and
smoothness terms; its learning objective and exogenous-deviation ranking are
otherwise unchanged~\citep{han2025root}.

\item \textbf{EasyRCA$^\dagger$.} EasyRCA requires anomalous nodes and their
onsets~\citep{assaad2023root}. A normal-only frontend scores continuous
level/difference tails and discrete state/transition surprise, calibrating each
variable to a 0.05 null activation rate on 2,048 held-out normal blocks. Because
EasyRCA relies on onset order, two consecutive activations define an onset so
that isolated threshold crossings do not create unstable ordering. EasyRCA
then uses the provided graph with temporal self-loops and run-contained
direct-effect regressions; ill-posed regressions abstain.

\item \textbf{T-RCA$^\dagger$.} The same frontend supplies thresholded event
trajectories to T-RCA's provided-graph branch~\citep{zan2024fly}. T-RCA reasons
over thresholded observations at each time point, so requiring persistence
would alter its input semantics; we therefore retain pointwise activation.
Deterministic ordering is added only when ranking metrics require an order.
\end{itemize}

The adapters use no root labels, operating-regime labels, or root cardinality. Event
windows contain only observations available by inference time.

\subsection{Baseline Settings}

Only non-default evaluation choices are listed.
\begin{itemize}
\setlength{\itemsep}{2pt}
\item \textbf{Candidate scope.} Every method is restricted to
\(\mathcal C\), the union of ancestors of the top-level active alarms
\(\mathcal A^\top\), as defined in the main paper.

\item \textbf{Native root set.} \(R_M^{\mathrm{nat}}\) denotes the root set
produced by method \(M\) before alarm-coverage completion. It is the Holm
rejection set at \(\alpha=0.01\) for ranking methods and the returned set for
set-valued methods.

\item \textbf{Fixed ranking.}
\(\pi_M=(r_1,\ldots,r_{|\mathcal C|})\) is method \(M\)'s fixed candidate
order. For A@1 and C@3, EasyRCA$^\dagger$ and T-RCA$^\dagger$ place native
roots before other anomalous ancestors and remaining candidates.

\item \textbf{Alarm coverage.} If a native set does not cover every top-level
active alarm, Algorithm~\ref{alg:supp-alarm-coverage} adds eligible ancestors
without changing the method ranking.
\end{itemize}

For each \(A_j\in\mathcal A^\top\), let
\(\mathcal C_j=\operatorname{an}(A_j)\) be its eligible ancestors. The
\emph{alarm-covered root set} \(R_M^{\mathrm{cov}}\) adds the smallest subset
of ranked candidates needed to intersect every \(\mathcal C_j\).

\begin{table}[!b]
\centering
\fontsize{9}{10.5}\selectfont
\setlength{\tabcolsep}{3.2mm}
\renewcommand{\arraystretch}{0.88}
\begin{tabular}{@{}lrr@{}}
\toprule
Method & Native ES & Alarm-covered ES \\
\midrule
MATERO-RCA      & \textbf{94.3} & \textbf{94.3} \\
StableRCA       & 48.5 & \textbf{70.9} \\
RCG             & 54.3 & \textbf{81.0} \\
CIRCA           & 26.4 & \textbf{30.5} \\
SmoothTraversal & 20.3 & \textbf{32.5} \\
IDI$^\dagger$   & 28.6 & \textbf{44.2} \\
AERCA           & 14.5 & \textbf{29.1} \\
AERCA$^\dagger$ & 13.4 & \textbf{29.2} \\
BARO            & 18.4 & \textbf{41.0} \\
EasyRCA$^\dagger$ & 22.1 & \textbf{33.6} \\
T-RCA$^\dagger$   & 35.2 & \textbf{44.3} \\
\bottomrule
\end{tabular}
\caption{Five-group ES (\%) before and after the same label-free alarm-coverage
completion. It improves uncovered baselines and applies MATERO-RCA's coverage
requirement uniformly for a fair comparison. Bold marks row-wise gains; ties
are retained.}
\label{tab:supp-shared-coverage}
\end{table}

\begin{algorithm}[H]
\footnotesize
\caption{Alarm-coverage completion}
\label{alg:supp-alarm-coverage}
\begin{algorithmic}[1]
\Require \(R_M^{\mathrm{nat}}\), ordered candidates
  \(\pi_M=(r_1,\ldots,r_{|\mathcal C|})\), and
  \(\{\mathcal C_j\}_{A_j\in\mathcal A^\top}\)
\Ensure Alarm-covered root set \(R_M^{\mathrm{cov}}\)
\State \(\mathcal U\gets
  \{j:R_M^{\mathrm{nat}}\cap\mathcal C_j=\varnothing\}\)
\If{\(\mathcal U=\varnothing\)}
  \State \Return \(R_M^{\mathrm{nat}}\)
\EndIf
\For{\(i=1,\ldots,|\mathcal C|\)}
  \State \(\mathcal V_i\gets\{j\in\mathcal U:r_i\in\mathcal C_j\}\)
\EndFor
\State \(\Delta_M\gets\) exact bit-mask minimum cover of
  \(\mathcal U\) by \(\{\mathcal V_i:r_i\notin R_M^{\mathrm{nat}}\}\)
\State Break ties by the summed and then ordered indices in \(\pi_M\)
\State \Return \(R_M^{\mathrm{cov}}\gets
  R_M^{\mathrm{nat}}\cup\Delta_M\)
\end{algorithmic}
\end{algorithm}

\(\Delta_M=R_M^{\mathrm{cov}}\setminus R_M^{\mathrm{nat}}\) contains only
label-free coverage additions; Holm's guarantee applies to
\(R_M^{\mathrm{nat}}\), not \(\Delta_M\).

\section{Additional Results}
\label{app:additional-results}

This section reports ranking diagnostics, causRCA setting breakdowns,
sensitivity analyses, and failure cases.

\subsection{causRCA Setting-Level Results}
\label{app:causrca-results}

Table~\ref{tab:supp-causrca-settings} disaggregates the causRCA macro over its
four settings and primary metrics. MATERO-RCA recovers all four settings
exactly, while competing methods degrade most on Probe and Full; the macro is
therefore not driven by the easier Coolant and Hydraulics views.

\subsection{Ranking and Selection Diagnostics}
\label{app:additional-metrics}

Table~\ref{tab:supp-extended-metrics} adds three higher-is-better ranking
metrics for each dataset group.

\noindent\textbf{MRR} measures the first annotated root; \textbf{MAP@3} and
\textbf{NDCG@3} measure multi-root quality in the top three, with linear and
logarithmic rank weighting, respectively.

\begin{figure*}[!t]
\begin{minipage}{\textwidth}
\begin{center}
\includegraphics[width=0.955\textwidth]{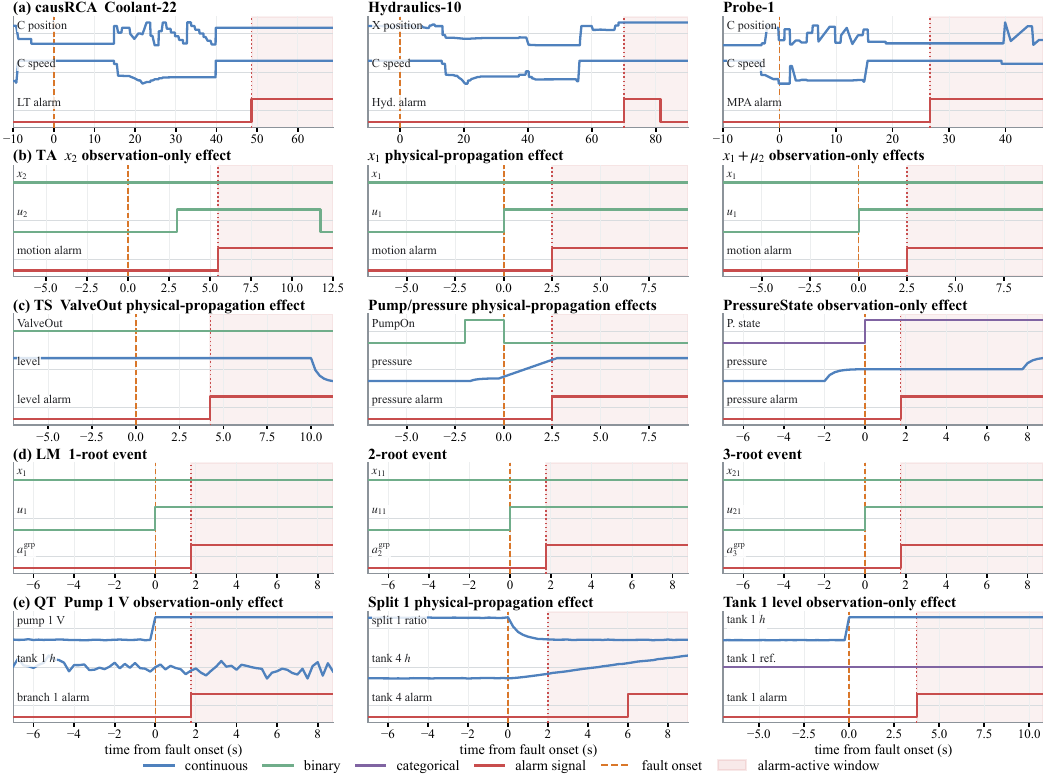}
\captionof{figure}{Representative event trajectories (three per dataset group).
Rows (a)--(e) show causRCA, TA, TS, LM, and QT; the causRCA examples span
Coolant, Hydraulics, and Probe. Each panel separates a root, response or
driver, and alarm trajectory. Orange and red dotted lines mark fault and
first-alarm onset, respectively, and red shading marks the post-alarm interval.
Tracks are independently min--max normalized and vertically offset, so
amplitudes are not comparable.}
\label{fig:supp-dataset-form}

{\small
\setlength{\tabcolsep}{0.55mm}
\renewcommand{\arraystretch}{0.96}
\resizebox{\textwidth}{!}{%
\begin{tabular}{@{}l*{6}{ccc}@{}}
\toprule
Method
& \multicolumn{3}{c}{causRCA}
& \multicolumn{3}{c}{TA}
& \multicolumn{3}{c}{TS}
& \multicolumn{3}{c}{LM}
& \multicolumn{3}{c}{QT}
& \multicolumn{3}{c}{Mean} \\
\cmidrule(lr){2-4}\cmidrule(lr){5-7}\cmidrule(lr){8-10}
\cmidrule(lr){11-13}\cmidrule(lr){14-16}\cmidrule(lr){17-19}
& MRR & M@3 & NDCG
& MRR & M@3 & NDCG
& MRR & M@3 & NDCG
& MRR & M@3 & NDCG
& MRR & M@3 & NDCG
& MRR & M@3 & NDCG \\
\midrule
EasyRCA$^\dagger$
& 91.1 & 84.2 & 85.1
& 22.0 & 4.4 & 6.7
& 75.3 & 70.6 & 74.0
& 24.5 & 9.4 & 11.7
& 73.2 & 66.0 & 71.5
& 57.2 & 46.9 & 49.8 \\
T-RCA$^\dagger$
& 91.1 & 84.2 & 85.1
& 22.0 & 4.4 & 6.7
& 73.1 & 68.3 & 72.4
& 24.5 & 9.4 & 11.7
& 97.7 & 95.7 & 96.0
& 61.7 & 52.4 & 54.4 \\
AERCA
& 78.8 & 74.0 & 76.7
& 27.4 & 15.6 & 18.3
& 67.2 & 57.2 & 59.6
& 6.3 & 0.0 & 0.0
& 65.3 & 58.6 & 64.0
& 49.0 & 41.1 & 43.7 \\
AERCA$^\dagger$
& 81.7 & 75.5 & 77.6
& 24.4 & 12.2 & 15.9
& 58.9 & 47.2 & 50.5
& 8.9 & 2.0 & 2.0
& 72.1 & 65.2 & 69.1
& 49.2 & 40.4 & 43.0 \\
CIRCA
& 88.7 & 88.9 & 91.8
& 61.2 & 51.7 & 57.6
& 69.7 & 65.8 & 70.4
& 51.4 & 31.7 & 37.1
& \underline{97.9} & \underline{97.1} & \underline{97.1}
& 73.8 & 67.1 & 70.8 \\
RCG
& \underline{96.7} & \underline{95.6} & \underline{97.0}
& \underline{89.4} & \underline{89.4} & \underline{92.1}
& \underline{96.7} & \underline{96.1} & \underline{97.3}
& \underline{98.0} & \underline{94.2} & \underline{96.1}
& 96.2 & 91.4 & 92.8
& \underline{95.4} & \underline{93.4} & \underline{95.1} \\
StableRCA
& 94.7 & 91.8 & 92.9
& 56.7 & 51.1 & 60.1
& \textbf{98.3} & \textbf{98.3} & \textbf{98.8}
& \textbf{100.0} & 92.3 & 94.9
& 93.3 & 91.9 & 94.3
& 88.6 & 85.1 & 88.2 \\
SmoothTraversal
& 80.0 & 75.9 & 79.2
& 33.0 & 24.4 & 33.6
& 75.6 & 70.0 & 72.7
& 51.9 & 33.9 & 39.0
& 54.8 & 48.1 & 55.3
& 59.0 & 50.5 & 56.0 \\
BARO
& 89.5 & 87.8 & 88.0
& 17.8 & 4.4 & 6.7
& 61.1 & 56.4 & 62.4
& 4.1 & 0.0 & 0.0
& 81.9 & 79.5 & 84.6
& 50.9 & 45.6 & 48.3 \\
IDI$^\dagger$
& 87.3 & 83.1 & 86.1
& 79.1 & 77.8 & 81.7
& 53.9 & 46.7 & 53.6
& 97.1 & 92.8 & 94.3
& 88.1 & 82.9 & 85.6
& 81.1 & 76.6 & 80.3 \\
\midrule
MATERO-RCA
& \textbf{100.0} & \textbf{100.0} & \textbf{100.0}
& \textbf{98.3} & \textbf{98.3} & \textbf{98.8}
& 91.7 & 91.9 & 94.1
& \textbf{100.0} & \textbf{100.0} & \textbf{100.0}
& \textbf{100.0} & \textbf{100.0} & \textbf{100.0}
& \textbf{98.0} & \textbf{98.1} & \textbf{98.6} \\
\bottomrule
\end{tabular}
}
}
\captionof{table}{Additional ranking results (\%). MRR measures how early the
first annotated root appears. M@3 denotes MAP@3, and NDCG is also evaluated at
3; both measure multi-root ranking quality. causRCA is the equal-weight macro over its four settings,
and Mean is the equal-weight macro over the five dataset groups. Best results
are bold and second-best results are underlined.}
\label{tab:supp-extended-metrics}
\end{center}
\end{minipage}
\end{figure*}

\begin{table*}[!t]
\begin{minipage}{\textwidth}
\begin{center}
{\small
\setlength{\tabcolsep}{0.55mm}
\renewcommand{\arraystretch}{0.96}
\resizebox{0.93\textwidth}{!}{%
\begin{tabular}{@{}l*{4}{rrrr}@{}}
\toprule
Method
& \multicolumn{4}{c}{Probe ($n=34$)}
& \multicolumn{4}{c}{Coolant ($n=25$)}
& \multicolumn{4}{c}{Hydraulics ($n=41$)}
& \multicolumn{4}{c}{Full ($n=100$)} \\
\cmidrule(lr){2-5}\cmidrule(lr){6-9}
\cmidrule(lr){10-13}\cmidrule(lr){14-17}
& \multicolumn{1}{c}{A@1} & \multicolumn{1}{c}{C@3}
& \multicolumn{1}{c}{Set F1} & \multicolumn{1}{c}{ES}
& \multicolumn{1}{c}{A@1} & \multicolumn{1}{c}{C@3}
& \multicolumn{1}{c}{Set F1} & \multicolumn{1}{c}{ES}
& \multicolumn{1}{c}{A@1} & \multicolumn{1}{c}{C@3}
& \multicolumn{1}{c}{Set F1} & \multicolumn{1}{c}{ES}
& \multicolumn{1}{c}{A@1} & \multicolumn{1}{c}{C@3}
& \multicolumn{1}{c}{Set F1} & \multicolumn{1}{c}{ES} \\
\midrule
StableRCA
& 79.4 & \underline{82.4} & 72.5 & \underline{58.8}
& \textbf{100.0} & \textbf{100.0} & \textbf{100.0} & \textbf{100.0}
& \textbf{100.0} & \textbf{100.0} & \textbf{100.0} & \textbf{100.0}
& 90.0 & 87.0 & 89.0 & \underline{85.0} \\
RCG
& \underline{82.4} & \textbf{100.0} & \underline{79.6} & 41.2
& \textbf{100.0} & \textbf{100.0} & \textbf{100.0} & \textbf{100.0}
& \textbf{100.0} & \textbf{100.0} & \underline{87.0} & 56.1
& \underline{92.0} & \textbf{100.0} & 88.2 & 62.0 \\
CIRCA
& 35.3 & \textbf{100.0} & 68.5 & 8.8
& \textbf{100.0} & \textbf{100.0} & \textbf{100.0} & \textbf{100.0}
& \textbf{100.0} & \textbf{100.0} & \textbf{100.0} & \textbf{100.0}
& 78.0 & \textbf{100.0} & \underline{89.3} & 69.0 \\
SmoothTraversal
& 2.9 & 55.9 & 2.9 & 2.9
& \textbf{100.0} & \textbf{100.0} & \textbf{100.0} & \textbf{100.0}
& \textbf{100.0} & \textbf{100.0} & \underline{87.0} & 56.1
& 67.0 & 85.0 & 61.7 & 49.0 \\
IDI$^\dagger$
& 64.7 & 70.6 & 61.2 & 29.4
& \textbf{100.0} & \textbf{100.0} & \textbf{100.0} & \textbf{100.0}
& 78.0 & \textbf{100.0} & 78.9 & 46.3
& 79.0 & \underline{89.0} & 77.2 & 51.0 \\
AERCA
& 23.5 & 32.4 & 27.5 & 8.8
& \textbf{100.0} & \textbf{100.0} & \textbf{100.0} & \textbf{100.0}
& 70.7 & \textbf{100.0} & 84.4 & 61.0
& 79.0 & 85.0 & 73.9 & 54.0 \\
AERCA$^\dagger$
& 23.5 & 32.4 & 19.1 & 11.8
& \textbf{100.0} & \textbf{100.0} & \textbf{100.0} & \textbf{100.0}
& \underline{90.2} & \textbf{100.0} & 86.3 & \underline{70.7}
& 84.0 & 79.0 & 75.5 & 67.0 \\
BARO
& 61.8 & 64.7 & 56.9 & 47.1
& \textbf{100.0} & \textbf{100.0} & \textbf{100.0} & \textbf{100.0}
& \textbf{100.0} & \textbf{100.0} & \textbf{100.0} & \textbf{100.0}
& 87.0 & \underline{89.0} & 85.0 & 81.0 \\
EasyRCA$^\dagger$
& 64.7 & 41.2 & 56.9 & 41.2
& \textbf{100.0} & \textbf{100.0} & \textbf{100.0} & \textbf{100.0}
& \textbf{100.0} & \textbf{100.0} & \underline{87.0} & 56.1
& 88.0 & 80.0 & 80.0 & 62.0 \\
T-RCA$^\dagger$
& 64.7 & 41.2 & 56.9 & 41.2
& \textbf{100.0} & \textbf{100.0} & \textbf{100.0} & \textbf{100.0}
& \textbf{100.0} & \textbf{100.0} & \underline{87.0} & 56.1
& 88.0 & 80.0 & 80.0 & 62.0 \\
\midrule
MATERO-RCA
& \textbf{100.0} & \textbf{100.0} & \textbf{100.0} & \textbf{100.0}
& \textbf{100.0} & \textbf{100.0} & \textbf{100.0} & \textbf{100.0}
& \textbf{100.0} & \textbf{100.0} & \textbf{100.0} & \textbf{100.0}
& \textbf{100.0} & \textbf{100.0} & \textbf{100.0} & \textbf{100.0} \\
\bottomrule
\end{tabular}
}
}
\captionof{table}{causRCA setting-level results under the shared evaluation protocol
(\%). Probe, Coolant, and Hydraulics are subsystem views; Full evaluates the
combined 92-node graph over the same 100 physical events. Best results are bold
and second-best results are underlined; ties are retained.}
\label{tab:supp-causrca-settings}
\end{center}
\end{minipage}
\end{table*}

Table~\ref{tab:supp-holm-alpha} varies Holm's \(\alpha\) while keeping candidate
\(p\)-values and rankings fixed. The main setting, \(\alpha=0.01\), is not
tuned per method or dataset.

\begin{table}[!t]
\centering
\fontsize{8}{9.5}\selectfont
\setlength{\tabcolsep}{1.0mm}
\begin{tabular}{@{}lcccc@{}}
\toprule
& \multicolumn{3}{c}{Holm} & BH \\
\cmidrule(lr){2-4}\cmidrule(l){5-5}
Method & \(\alpha=0.01\) & \(\alpha=0.05\) & \(\alpha=0.10\) & \(q_{\rm BH}=0.05\) \\
\midrule
StableRCA       & 48.5 / \textbf{70.9} & 55.2 / 59.3 & \textbf{55.3} / 57.5 & 54.7 / 57.5 \\
RCG             & 54.3 / \textbf{81.0} & \textbf{63.9} / 69.2 & 60.6 / 65.3 & 62.3 / 67.1 \\
CIRCA           & 26.4 / \textbf{30.5} & \textbf{27.4} / 29.9 & \textbf{27.4} / 29.9 & 16.2 / 16.2 \\
SmoothTraversal & 20.3 / \textbf{32.5} & \textbf{21.8} / \textbf{32.5} & \textbf{21.8} / 31.9 & \textbf{21.8} / \textbf{32.5} \\
IDI$^\dagger$   & 28.6 / \textbf{44.2} & \textbf{29.2} / 30.3 & 27.5 / 28.6 & 23.7 / 27.1 \\
AERCA           & \textbf{14.5} / \textbf{29.1} & 13.9 / 28.4 & 12.6 / 27.6 & 12.6 / 27.0 \\
AERCA$^\dagger$ & \textbf{13.4} / \textbf{29.2} & 12.3 / 28.0 & 11.2 / 26.9 & 10.5 / 26.3 \\
BARO            & 18.4 / 41.0 & \textbf{21.3} / \textbf{41.5} & \textbf{21.3} / 40.9 & \textbf{21.3} / \textbf{41.5} \\
\midrule
Mean            & 28.1 / \textbf{44.8} & \textbf{30.6} / 39.9 & 29.7 / 38.6 & 27.9 / 36.9 \\
\bottomrule
\end{tabular}
\caption{Holm-\(\alpha\) and BH sensitivity for the eight ranking baselines
(five-group ES, \%). Each cell reports native / alarm-covered ES; the BH column
uses \(q_{\rm BH}=0.05\), and Holm \(\alpha=0.01\) is the main setting. Bold marks the
best value of each metric within a row; ties are retained.}
\label{tab:supp-holm-alpha}
\end{table}

Holm \(\alpha=0.01\) gives the highest mean alarm-covered ES (44.8\%) and is
best or tied for seven of eight baselines. Although \(\alpha=0.05\) improves
native ES, it is weaker after the shared coverage step; BH is weaker on both
means. Thus \(\alpha=0.01\) is the strongest common covered-set setting, not a
method-specific choice.

\subsection{Calibration Sensitivity}
\label{app:calibration-sensitivity}

Let \(\rho_C\) and \(\rho_A\) denote the empirical quantile levels used to set
\(\tau_i^C\) and \(\tau_j^A\) in Eq.~(5) of the main paper. We vary one at a
time with fixed checkpoints, event windows, and solver settings
(\(K=3,\gamma=0.25\)). Table~
\ref{tab:supp-calibration-sensitivity} reports the five-group macro.

\begin{table}[!t]
\centering
\small
\setlength{\tabcolsep}{2.7pt}
\begin{tabular}{@{}lcccccc@{}}
\toprule
Setting & \(\rho_C\) & \(\rho_A\) & A@1 & C@3 & Set F1 & ES \\
\midrule
Default & 0.99 & 0.95 & 96.0 & \textbf{100.0} & 97.8 & 94.3 \\
Compatibility & 0.95 & 0.95 & 94.7 & \textbf{100.0} & 97.5 & 93.0 \\
Alarm (lower) & 0.99 & 0.92 & \textbf{96.7} & \textbf{100.0} &
\textbf{98.5} & \textbf{95.0} \\
Alarm (higher) & 0.99 & 0.97 & 94.6 & 98.7 & 96.0 & 91.8 \\
\bottomrule
\end{tabular}
\caption{Frozen-checkpoint calibration sensitivity (\%). Bold marks the best
value in each metric column; ties are retained.}
\label{tab:supp-calibration-sensitivity}
\end{table}

Compatibility calibration is stable: changing \(\rho_C\) from 0.99 to 0.95
changes Set F1 by \(-0.3\) points. Alarm calibration is less invariant.
\(\rho_A=0.92\) changes only TA event 30, whereas \(\rho_A=0.97\) loses 1.8 Set-F1
points; at 0.99, 36/345 events have no alarm
energy. The prespecified \(\rho_A=0.95\) is therefore a balanced common setting,
not a test-set optimum.

\subsection{Discretization Sensitivity}
\label{app:discretization-sensitivity}

We vary the response-bin count \(B\in\{5,6,7,8\}\) for QT. Each continuous
response uses \(B-2\) equal-width intervals over its normal-training range plus
two outer bins; command and reference signals retain their native finite states.
Each variant is trained for its binning scheme. We set the compatibility
quantile to \(\rho_C=0.97\) for all variants while keeping the architecture, alarm
calibration, events, and solver settings fixed.

\begin{table}[!t]
\centering
\small
\setlength{\tabcolsep}{2.7pt}
\begin{tabular}{@{}cccccc@{}}
\toprule
Bins & A@1 & C@3 & Set F1 & ES & Evals. \\
\midrule
5 & 94.3 & 94.3 & 91.9 & 88.6 & 2.86 \\
6 & 94.3 & 94.3 & 91.8 & 85.7 & 3.20 \\
7 & \textbf{100.0} & \textbf{100.0} & \textbf{97.5} & \textbf{91.4} & \textbf{2.57} \\
8 & 94.3 & 94.3 & 91.8 & 85.7 & 2.86 \\
\bottomrule
\end{tabular}
\caption{QT response discretization sensitivity (\%). A@1 and C@3 denote
AnyRoot@1 and CompleteRoots@3; Evals.\ denotes mean root--mode evaluations per
event.}
\label{tab:supp-bin-sensitivity}
\end{table}

Seven bins is the only setting with perfect A@1 and C@3 and gives the highest
Set F1 (97.5\%) and ES (91.4\%). The other bin counts introduce only two
localization errors (events 16 and 20), both due to boundary placement. A
fault-induced change from \(x^{\mathrm{clean}}\) to \(x^{\mathrm{fault}}\) is
hidden when both values receive
the same discrete state,
\mbox{\(\mathsf Q_i^{(B)}(x^{\mathrm{clean}})
=\mathsf Q_i^{(B)}(x^{\mathrm{fault}})\)}. For example, in event 16,
\textit{Tank1Level} changes from approximately 12.04 to 13.05; the 5-, 6-, and 8-bin
maps place both values in the same bin, whereas the 7-bin map separates them.
Our equal-width bins use no domain knowledge. In practice, alarm limits,
deadbands, and operating regimes provide engineering priors for placing
boundaries \citep{weng2022evidence}, reducing such fault-relevant aliasing.

\subsection{Causal-Graph Sensitivity}
\label{app:graph-sensitivity}

The default uses the unperturbed QT graph. At each perturbation level, we
evaluate 10 fixed graph perturbations that remove three of 22
mechanism edges or add three/five acyclic edges, giving nominal \(-10\%\),
\(+10\%\), and \(+20\%\) perturbations.
Alarm edges and declared root-to-alarm reachability remain fixed, and the
\(+10\%\) additions are nested within \(+20\%\). We retrain and recalibrate
each graph with model seed 1 fixed, \(\rho_C=0.97\), seven bins, and \(K=3\).

The default yields \(100.0\%\) A@1/C@3, \(97.5\%\) Set F1, and \(91.4\%\) ES.
Across the 10 seeds, deletion lowers mean A@1 to \(96.6\%\), Set F1 to
\(94.1\%\), and ES to \(88.0\%\), while C@3 remains perfect. Its 12 additional
errors all replace the true
\textit{Tank1Level} root with the child of a removed edge. This reflects
omitted-edge misspecification: deleting a true relation removes both a
conditioning dependency and its physical propagation path, shifting the
explanation from the upstream cause to the affected child.

Additions preserve A@1/C@3 but mainly induce conservative over-selection.
An incorrect edge makes CompatNet learn a normal association between a child
and a noncausal parent. During a fault, the child may change while this added
parent does not, making a valid causal response appear unlikely and receive
excessive local energy. Conversely, an accidentally matching parent can mask
a true incompatibility. These errors can respectively block valid downstream
propagation or suppress the propagating root-effect mode \(\mathrm p\). All 27
new supersets follow this mechanism, with their extra roots drawn from
descendants lost from the default propagation closure. At \(+20\%\), mean Set F1 and ES are \(96.2\%\) and
\(86.9\%\). Overall, the effect is modest: mean Set F1 and ES remain at least
\(94.1\%\) and \(86.9\%\), respectively, across all perturbation levels.
Thus the sensitivity originates in local compatibility modeling
rather than search or ranking.

\begin{figure}[t]
  \centering
  \includegraphics[width=\columnwidth]{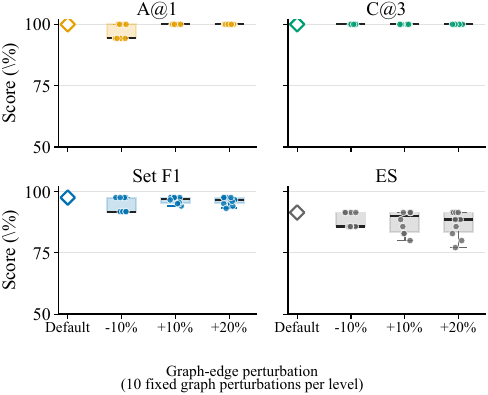}
  \caption{QT causal-graph sensitivity on 35 events. Each box and its points
  summarize \(n=10\) fixed graph perturbations at that level. The diamond
  is the separately evaluated default graph and is excluded from the boxes;
  model seed 1 is fixed throughout.}
  \label{fig:supp-graph-sensitivity}
\end{figure}

\FloatBarrier

\begin{figure}[!t]
  \centering
  \includegraphics[width=\columnwidth]{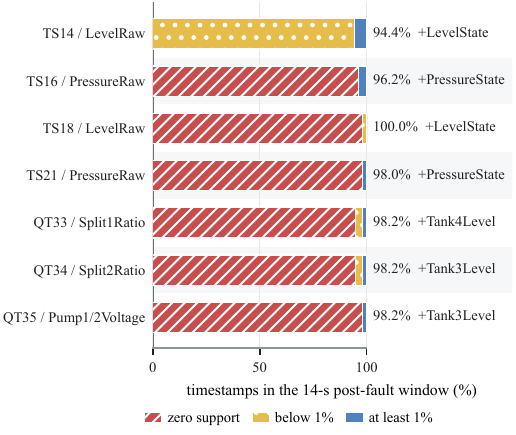}
  \caption{Normal-support composition over the 14-s post-fault window for seven
  OOD failures. For multi-root events, each timestamp uses the minimum
  conditional frequency across annotated root relations. Red, yellow, and blue
  indicate zero, below-\(1\%\), and at-least-\(1\%\) support; labels report the
  low-support share and added downstream root.}
  \label{fig:supp-ood-failures}
\end{figure}

\subsection{Failure Analysis}
\label{app:failure-cases}

Of the nine ES errors, seven are predefined OOD physical-propagation faults. We define
these as events whose annotated root--parent configurations have negligible
support in normal training, placing their downstream propagation outside the
mechanisms learned from normal data.

\noindent\textbf{OOD physical-propagation faults.}
Figure~\ref{fig:supp-ood-failures} measures the normal-training support of each
annotated root conditioned on its causal parents. All seven failures are
dominated by unseen or below-\(1\%\) contexts, confirming that their physical
propagation requires extrapolation beyond learned normal mechanisms.

Repairing only the annotated roots cannot reproduce these unsupported
downstream responses, leaving local-compatibility residuals. The objective
therefore retains every annotated root but adds one downstream variable to
absorb the unexplained residual, producing conservative supersets. For
example, TS event 14 retains \(\mathrm{LevelRaw}\) but adds its child
\(\mathrm{LevelState}\), while QT event 34 retains
\(\mathrm{Split2Ratio}\) but adds its child \(\mathrm{Tank3Level}\).

\noindent\textbf{Other cases.}
TA event 30 is a parent--child ambiguity: the returned parent \(u_1\) and
annotated child \(x_1\) attain
\mbox{\(\widehat U_{\{u_1\}}=\widehat U_{\{x_1\}}=0.25\)}. TS event 2 is a
discretization-boundary superset that retains the annotated
\(\mathrm{PumpOn}\) and \(\mathrm{PressureRaw}\) roots but adds the latter's
categorical child \(\mathrm{PressureState}\).

\begingroup
\fontsize{8.5}{9.5}\selectfont
\bibliography{aaai2027}
\endgroup

\end{document}